\documentclass{article}

\usepackage[T1]{fontenc}
\usepackage{amsmath,amssymb,latexsym,amsthm,enumerate,epsfig,subfigure}
\usepackage{booktabs}

\newcommand{\sgn}{{\rm sgn}}
\newcommand{\R}{\mathbb{R}}

\newtheorem{theorem}{Theorem}[section]

\newtheorem{lemma}[theorem]{Lemma}
\newtheorem{proposition}{Proposition}
\newtheorem{remark}{Remark}

\newcommand{\acal}{{\mathcal A}}

\begin{document}

\title{On a continuous mixed strategies model for evolutionary game theory}

\author{A.~Boccabella\thanks{Dipartimento di Scienze di Base e Applicate per l'Ingegneria
(SBAI), Sezione di Matematica, Facolt\`a di Ingegneria,
Universit\`a degli Studi \lq\lq Sapienza'' di Roma,
(boccabella@dmmm.uniroma1.it)}\and R.~Natalini\thanks{Istituto per
le Applicazioni del Calcolo \lq\lq M. Picone'', Consiglio
Nazionale delle Ricerche (roberto.natalini@cnr.it)} \and L.
Pareschi\thanks{$^3$ Dipartimento di Matematica \& CMCS,
Universit\`a degli Studi di Ferrara (lorenzo.pareschi@unife.it)}}

\date{September 15, 2010}
\maketitle

\begin{abstract}
We consider an integro-differential model for evolutionary game
theory which describes the evolution of a population adopting
mixed  strategies. Using a reformulation based on the first
moments of the solution, we prove some analytical properties of
the model and global estimates. The asymptotic behavior and the
stability of solutions in the case of two strategies is analyzed
in details. Numerical schemes for two and three strategies which
are able to capture the correct equilibrium states are also
proposed together with several numerical examples.
\end{abstract}



\section{Introduction}

Evolutionary dynamics is based on the ideas of mathematical game theory.
In game theory, a player's strategy in a game is a complete plan
of action at any stage of the game.
A \textit{pure strategy} defines a specific move or action
that a player will follow in every possible attainable situation
in a game.  A player's strategy set is the set of
pure strategies available to that player and defines what
strategies are available  to play. A \textit{mixed
strategy} is an assignment of a probability to each pure strategy.
This allows for a player to select a pure strategy with a given
distribution of probability. Since probabilities are continuous,
there are infinitely many mixed strategies available to a player,
even if their strategy set is finite. Of course, one can regard a
pure strategy as a degenerate case of a mixed strategy, in which
that particular pure strategy is selected with probability $1$ and
every other strategy with probability $0$.

In any game, an important concept is the {\textit{payoff}} that is the number
which represents the motivations of a player. The exact
definition of the payoff depends on the case of interest: payoff
may represent profit, utility, or other continuous measures, or
may simply rank the desirability of outcomes. In all cases, the
payoffs must reflect the motivations of the players. Following the
basic tenet of Darwinism, we may express the success of a player
in a game, that means the player's survival, as the difference
between the player's payoff and the average payoff of all players.

Dynamic models for continuous strategy spaces have received considerable
attention recently both in theoretical biology when considering the
evolution of species traits \cite{abrams, desvillettes, wennberg}
and in economy when predicting rational behavior of individuals whose
payoffs are given through game interactions \cite{marsili, friedman}.

In the present paper we analyze a continuous mixed strategies
model for population dynamics based on an integro-differential
representation. Analogous models based on the replicator equation
with continuous strategy space were recently investigated in
\cite{bomze, cressman, hofbauer, norman, oechssler}. In contrast
with finite strategy spaces, where the notion of equilibrium is
well understood and studied \cite{HofbauerSigmund, Weibull}, the
situation of games with infinite strategies is still missing a
general theory due to several technical and conceptual
difficulties \cite{norman}.

The model here considered is characterized by a continuous density
function $f(t,\bf q)$ of population adopting the $\bf
q\in\mathbb{R}^N$ strategy at time $t$ and presents some analogies
with classical kinetic or mean field approaches. In particular we
show that the model, which contains a cubic nonlinearity in $f$,
can be reformulated in terms of the first moments of the solution.
Such reformulation is essential in our analysis and in the
derivation of numerical approximations.

For the moment based model we prove global existence of solutions
and study the asymptotic behavior and stability of solutions in
the case of two strategies. Two classes of stationary solutions
are found. Continuous stationary solutions are characterized by
every density function with a given mean strategy. If we consider
more general solutions, so that the probability distributions are
no more absolutely continuous with respect to the Lebesgue
measure, another class of stationary solutions is given by
concentrated Dirac masses. Numerical schemes for the two and three
strategy case which are able to capture the correct equilibrium
states are also proposed together with several numerical examples.

The rest of the paper is organized as follows. In Section 2, we
present the model for $N$ pure strategies and prove a priori
estimates and the global existence of solutions. In Section 3, we
put the emphasis on the model with two pure strategies, which can
be reduced to a 1D model, and study the asymptotic behavior of the
solutions and their relation with stationary solutions. Section 4
is dedicated to the numerical approximation of the 1D model and to
numerical tests for the Prisoner's Dilemma and for the Hawk or
Dove games, with results about the a priori estimate, the
asymptotic behavior of the solutions and the stationary solutions.
In Section 5 and 6, we present the 2D model and the numerical
tests for the Rock-Scissors-Paper game. Some final considerations
are reported in the last section.

\section{An integro-differential model for continuous mixed strategies}
\subsection{Setting of the model}

First, we introduce an integro-differential model for
continuous mixed strategies. We start from some preliminary
concepts and definitions taken from \cite{HofbauerSigmund}.
Assume that we have a game where there are $N$ pure strategies
$R_1$ to $R_N$ and that the players can use mixed strategies:
these consists in playing the pure strategies $R_1$ to $R_N$ with
some probabilities $q_1$ to $q_N$ with $q_i\geq 0$ and $\sum
q_i=1$. A strategy corresponds to a point ${\bf{q}}$ in the
simplex
\begin{equation}\label{eq:simplex}
\mathcal{S}_N:=\displaystyle \lbrace {\bf{q}}=(q_1,\dots, q_N)\in \R^{N}:q_i\geq 0 \mbox{ and } \displaystyle \sum_{i=1}^N q_i=1\rbrace.
\end{equation}
The corners of the simplex are the standard unit vectors
${\bf{e}}_i$ with the $i$-th component is $1$ and all others are
$0$ and correspond to the $N$ pure strategies $R_i$,
$i=1,\dots,N$.

Let us denote by $a_{ij}$ the payoff for a player using the pure
strategy $R_i$ against a player using the pure strategy $R_j$. The
$N\times N$ matrix $\mathcal{A}=(a_{ij})$ is said to be the
\textit{payoff matrix}. An $R_i$-strategist obtains the expected
payoff $(\mathcal{A}{\bf{q^{*}}})_i=\sum_j a_{ij}q^{*}_j$ against
a ${\bf{q^{*}}}$-strategist, since $q^{*}_j$ is the probability
that he is met with strategy $R_j$. The payoff for a
${\bf{q}}$-strategist against a ${\bf{q^{*}}}$-strategist is given
by
\begin{equation}\label{eq:payoff}
\displaystyle A({\bf{q}},{\bf{q^{*}}}):={\bf{q}}\cdot \mathcal{A}{\bf{q^{*}}}=\sum_{i,j=1}^Na_{ij}q_iq^{*}_j.
\end{equation}

We consider a population of individuals as a player of the game
and denote by $f(t,{\bf{q}})$ the density of population adopting
the ${\bf{q}}$ strategy at time $t$; the evolution in time of $f$,
due to the dynamics of the game, is driven by
\begin{equation}\label{eq:fpN}
\displaystyle \partial_t
f(t,{\bf{q}})=f(t,{\bf{q}})\,\left(\int_{\mathcal{S}_{N}}
A({\bf{q}},{\bf{q^{*}}})f(t,{\bf{q^{*}}})\,d{\bf{q^{*}}}-\phi(f)\right),
\end{equation} where the term
\begin{equation}\label{eqpayoff}
\int_{\mathcal{S}_{N}} A({\bf{q}},{\bf{q^{*}}})f(t,{\bf{q^{*}}})\,d{\bf{q^{*}}}
\end{equation}
represents the payoff of the strategy ${\bf{q}}$ against all the
others strategies, $A({\bf{q}},{\bf{q^{*}}})$ being the
interacting kernel between the ${\bf{q}}$-strategist and the
${\bf{q^{*}}}$-strategist. The last term of the equation
(\ref{eq:fp}) is defined by
\begin{equation}\label{eq:phi}
\phi(f):=\int_{\mathcal{S}_{N}}\,\int_{\mathcal{S}_{N}} f(t,{\bf{q}})\,A({\bf{q}},{\bf{q^{*}}})\,f(t,{\bf{q^{*}}})\,d{\bf{q^{*}}}\,d{\bf{q}}
\end{equation}
and represents the average payoff of the population.\\

Since $\sum_{i=1}^{N}q_i=1,$ we can reduce the number of
variables, considering $$q_N=1-\sum_{i=1}^{N-1}q_i$$ and obtaining
the $(N-1)$ - dimensional model (\ref{eq:fpN}), on the simplex
\begin{equation}\label{eq:newset}
\mathcal{T}_{N-1}:=\lbrace {\bf{p}}= (p_1,p_2,\dots,p_{N-1})\in \R^{N-1}\,|\,p_i\geq 0\,,\sum_{i=1}^{N-1}p_i\leq 1\rbrace,
\end{equation}
namely
\begin{equation}\label{eq:fp}
\displaystyle
\partial_t f(t,{\bf{p}})=f(t,{\bf{p}})\,\left(\int_{\mathcal{T}_{N-1}}
A({\bf{p}},{\bf{p^{*}}})f(t,{\bf{p^{*}}})\,d{\bf{p^{*}}}-\phi(f)\right),
\end{equation}
with $A({\bf{p}},{\bf{p^{*}}})$ defined by
\begin{equation}\label{AT}
\displaystyle A({\bf{p}},{\bf{p^{*}}}):={\bf{p}}\cdot \mathcal{A}{\bf{p^{*}}}=\sum_{i,j=1}^{N-1}a_{ij}p_ip^{*}_j,
\end{equation}
and $\phi$ defined by
\begin{equation}\label{eq:newphi}
\phi(f):=\int_{\mathcal{T}_{N-1}}\,\int_{\mathcal{T}_{N-1}} f(t,{\bf{p}})\,A({\bf{p}},{\bf{p^{*}}})\,f(t,{\bf{p^{*}}})\,d{\bf{p^{*}}}\,d{\bf{p}}.\\
\end{equation}
\begin{remark}
If we take an initial condition
\begin{equation}\label{init}
f(0,{\bf{p}})=f_0({\bf{p}})\geq 0,
\end{equation}
with $\int_{\mathcal{T}_{N-1}} f_0({\bf{p}})d{\bf{p}}=1$, then it
is easy to see that $f\geq 0$ for all $t> 0$ and if
$f_0({\bf{\bar{p}}})=0$ for some ${\bf{\bar{p}}}$, then
$f(t,{\bf{\bar{p}}})=0$ for all $t>0$. We have also that
\begin{equation}
\int_{\mathcal{T}_{N-1}} f(t,{\bf{p}})d{\bf{p}}=1, \quad \forall\, t> 0.\label{massaf}
\end{equation}
This follows from the mass conservation, by integrating the
equation (\ref{eq:fp}) w.r.t. ${\bf{p}}$ and using (\ref{eq:newphi})
and (\ref{massaf})
\begin{equation}\label{eq:massconservation}
\partial_t \displaystyle \int_{\mathcal{T}_{N-1}} f(t,{\bf{p}})\,d{\bf{p}}=0.
\end{equation}\\
\end{remark}
Let us introduce the moments for $f$:
\begin{equation}\label{moments}
M_{\bf{k}}(f):=\int_{\mathcal{T}_{N-1}}{\bf{p}}^{{\bf{k}}}f({\bf{p}})\,d{\bf{p}}=\int_{\mathcal{T}_{N-1}}p_1^{k_1}\,p_2^{k_2}\,\dots\,p_{N-1}^{k_{N-1}}\,f({\bf{p}})\,d{\bf{p}},
\end{equation}
with ${\bf{k}}:=(k_1,k_2,\dots,k_{N-1})$. Using $M_{\bf{k}}(f)$, the payoff and the average payoff (\ref{eq:newphi}) are expressed respectively by
\begin{equation}\label{eq:Apayoffmoments}
\int_{\mathcal{T}_{N-1}}A({\bf{p}},{\bf{p^*}})\,f(t,{\bf{p^{*}}})\,d{\bf{p^{*}}}=\sum_{j=1}^{N-1}M_{{\bf{e}}_j}(f)\,\left(\sum_{i=1}^{N-1}\vartheta_{i,j}\,p_i+\varsigma_j\right)+a_{N,N}+\sum_{i=1}^{N-1}\upsilon_i\,p_i,
\end{equation}
\begin{equation}\label{eq:phimoments}
\phi (f)=\sum_{j=1}^{N-1}M_{{\bf{e}}_j}(f)\,\left(\sum_{i=1}^{N-1}\vartheta_{i,j}M_{{\bf{e}}_i}(f)+\varsigma_j\right)+a_{N,N}+\sum_{i=1}^{N-1}\upsilon_i\,M_{{\bf{e}}_i}(f),
\end{equation}
where ${\bf{e}}_i\in \R^{N-1}$ is the standard unit vector with
the $i$-th component equal to $1$ and all others equal to $0$,
$\vartheta_{i,j}:=a_{i,j}-a_{i,N}-a_{N,j}+a_{N,N}$,
$\varsigma_j:=a_{N,j}-a_{N,N}$, $\upsilon_i:=a_{i,N}-a_{N,N}$.\\\\
In the final form of the equation (\ref{eq:fp}), that we will use
later in this paper, the only integral terms are the first moments
$M_{{\bf{e}}_i}$: \\
\begin{equation}\label{eq:fpmoments}
\displaystyle \partial_t f(t,{\bf{p}})=f(t,{\bf{p}})\,\left(\sum_{i=1}^{N-1} (p_i-M_{{\bf{e}}_i}(f))\left(\upsilon_i+\sum_{j=1}^{N-1}\vartheta_{i,j}\,M_{{\bf{e}}_j}(f)\right)\right).
\end{equation}
\vspace{0.3cm}
\subsection{Global existence of the solutions}
We consider the Cauchy problem (\ref{eq:fpmoments})-(\ref{init}) for $t\geq 0$ and ${\bf{p}}\in \mathcal{T}_{N-1}$, i.e.
\begin{small}
\begin{equation}\label{CauchyNdim}
\begin{cases}
\displaystyle \partial_t f(t,{\bf{p}})=f(t,{\bf{p}})\,\left(\sum_{i=1}^{N-1} (p_i-M_{{\bf{e}}_i}(f))\left(\upsilon_i+\sum_{j=1}^{N-1}\vartheta_{i,j}\,M_{{\bf{e}}_j}(f)\right)\right)\\
f(0,{\bf{p}})=f_0({\bf{p}}),\\
\end{cases}
\end{equation}
\end{small}
with $f_0({\bf{p}})\geq 0$ and $\int_{\mathcal{T}_{N-1}}f_0({\bf{p}})d{\bf{p}}=1$.
\begin{proposition}[Local existence]\label{Localexistence}
For all $M>0$ there exists $T(M)>0$ such that if
$||f_0({\bf{p}})||\leq M$, then there exists a unique solution
$f\in C([0,\tilde{T}]\times \mathcal{T}_{N-1})$ for the problem
(\ref{CauchyNdim}), for all $\tilde{T}\leq T(M)$.
\end{proposition}
\begin{proof}
Let us define
\begin{equation}\label{eq:time}
T(M):= \max_{R\geq M}\, \min (T_1(R,M),T_2(R,M)),
\end{equation}
where
\begin{eqnarray*}\label{eq:appfunctions}
&T_1(R,M)&:=\dfrac{R}{(R+M)(1+R+M)(V+\Theta(R+M))},\\
&T_2(R,M)&:=\dfrac{1}{S(R,M)},
\end{eqnarray*}
and
\begin{small}
\begin{equation*}\label{fappoggio}
S(R,M):=(1+R+M)(V+\Theta(R+M))+(R+M)(V+\Theta(1+2(R+M)))+2\Theta(R+M)^2,
\end{equation*}
\end{small}
\begin{equation}\label{eq:Vtheta}
V:=\sum_{i=1}^{N-1}|\upsilon_i|, \quad \mbox{and} \quad \Theta :=\sum_{i=1}^{N-1}\sum_{j=1}^{N-1}|\vartheta_{i,j}|.
\end{equation}
We define the set
\begin{equation}\label{eq:setBR}
B_R:=\lbrace g \in C([0,T]\times \mathcal{T}_{N-1})\,|\, |g-f_0|_C\leq R\rbrace,
\end{equation}
for $R\geq M$, and, for all $\, g\in B_R$, the operator
\begin{small}
\begin{equation}\label{eq:operatore}
G(g)({\bf{p}}):=f_0({\bf{p}})+\int_0^t g({\bf{p}})\,\left(\sum_{i=1}^{N-1} (p_i-M_{{\bf{e}}_i}(g))\left(\upsilon_i+\sum_{j=1}^{N-1}\vartheta_{i,j}\,M_{{\bf{e}}_j}(g)\right)\right)\,dt.
\end{equation}
\end{small}
We have that for all $\, g\in B_R,$
$$|M_{{\bf{e}}_k}(g)|\leq R+M,\qquad \forall\, k,$$
$$|p_k-M_{{\bf{e}}_k}(g)|\leq 1+R+M,\qquad \forall\, k.$$
$$$$
It is easy to prove that $G(g)\in B_R$ for $t\leq T(M)$:
\begin{small}
\begin{eqnarray*}
|G(g)-f_0| \leq \int_0^t |g-f_0+f_0|\,\left(\sum_{i=1}^{N-1} |p_i-M_{{\bf{e}}_i}(g)|\left(|\upsilon_i|+\sum_{j=1}^{N-1}|\vartheta_{i,j}|\,|M_{{\bf{e}}_j}(g)|\right)\right)\,dt\\
\leq \int_0^t (R+M)\,\left(\sum_{i=1}^{N-1} (1+R+M)\left(|\upsilon_i|+(R+M)\sum_{j=1}^{N-1}|\vartheta_{i,j}|\right)\right)\,dt\\\\
=t (R+M)(1+R+M)(V+\Theta(R+M))\leq R.
\end{eqnarray*}
\end{small}\\

The operator $G(g)$ is a contraction on $B_R$ for $t\leq T(M)$: for all $\, g,\, \tilde{g} \in B_R$\\
\begin{small}
\begin{eqnarray*}
|G(g)-G(\tilde{g})|&=&\\
&=&\vert \int_0^t  \left[g\,\left(\sum_{i=1}^{N-1} (p_i-M_{{\bf{e}}_i}(g))\left(\upsilon_i+\sum_{j=1}^{N-1}\vartheta_{i,j}\,M_{{\bf{e}}_j}(g)\right)\right)\right]dt \\
&-&\int_0^t \left[\tilde{g}\left(\sum_{i=1}^{N-1} (p_i-M_{{\bf{e}}_i}(\tilde{g}))\left(\upsilon_i+\sum_{j=1}^{N-1}\vartheta_{i,j}\,M_{{\bf{e}}_j}(\tilde{g})\right)\right)\right] \, dt\vert \\
&\leq & \int_0^t  |g-\tilde{g}|\left(\sum_{i=1}^{N-1} |p_i-M_{{\bf{e}}_i}(g)|\left(|\upsilon_i+\sum_{j=1}^{N-1}|\vartheta_{i,j}|\,|M_{{\bf{e}}_j}(g)|\right)\right)\,dt
\\&+&\int_0^t  |\tilde{g}|\left[ \left(\sum_{i=1}^{N-1} |p_i-M_{{\bf{e}}_i}(g)|\left(|\upsilon_i|+\sum_{j=1}^{N-1}|\vartheta_{i,j}|\,|M_{{\bf{e}}_j}(g)|\right)\right)\right]dt
\end{eqnarray*}
\begin{equation*}
-\int_0^t |\tilde{g}|\left[\left(\sum_{i=1}^{N-1} |p_i-M_{{\bf{e}}_i}(\tilde{g})|\left(|\upsilon_i|+\sum_{j=1}^{N-1}|\vartheta_{i,j}|\,|M_{{\bf{e}}_j}(\tilde{g})|\right)\right) \right] dt
\end{equation*}
\begin{equation*}
= \int_0^t  |g-\tilde{g}|\left(\sum_{i=1}^{N-1} |p_i-M_{{\bf{e}}_i}(g)|\left(|\upsilon_i|+\sum_{j=1}^{N-1}|\vartheta_{i,j}|\,|M_{{\bf{e}}_j}(g)|\right)\right)\,dt
\end{equation*}
\begin{equation*}
+\int_0^t  |\tilde{g}|\left(\sum_{i=1}^{N-1}|\upsilon_i||M_{{\bf{e}}_i}(g)-M_{{\bf{e}}_i}(\tilde{g})|\right)dt
\end{equation*}
\begin{equation*}
+\int_0^t |\tilde{g}|\left( \sum_{i=1}^{N-1}\left(\sum_{j=1}^{N-1}|\vartheta_{i,j}|\,|M_{{\bf{e}}_j}(g)-M_{{\bf{e}}_j}(\tilde{g})||p_i-(M_{{\bf{e}}_i}(g)+M_{{\bf{e}}_i}(\tilde{g}))|\right)\right) dt
\end{equation*}\end{small}
\begin{equation*}
\leq t\,S(R,M) \, sup |g-\tilde{g}|.
\end{equation*}
\\
The last inequality is obtained using the following inequalities, for all $\, g,\, \tilde{g} \in B_R$:
$$|\tilde{g}|\leq R+M,$$
$$|M_{{\bf{e}}_k}(g)-M_{{\bf{e}}_k}(\tilde{g})|\leq sup|g-\tilde{g}|,\qquad \forall\, k$$
$$|M_{{\bf{e}}_{k_1}}(g)M_{{\bf{e}}_{k_2}}(\tilde{g})-M_{{\bf{e}}_{k_2}}(g)M_{{\bf{e}}_{k_1}}(\tilde{g})|\leq 2(R+M)sup|g-\tilde{g}|\qquad \forall\, k_1,\,k_2.$$\\
We have that $G(g)$ is a contraction on $B_R$ for all $t\leq T(M)$
and so problem (\ref{CauchyNdim}) admits a unique solution $f\in
C([0,\tilde{T}]\times \mathcal{T}_{N-1})$, for all $\tilde{T}\leq
T(M)$.
\end{proof}
We proved the local existence of solution in a time interval
$(0,\tilde{T})$, depending on $M$. Now we define $T_{\max}$ as the
time limit in which this local solution exists.
\begin{lemma}\label{limiteinfinito}
If $T_{\max}<+\infty$ then $\limsup_ {t\rightarrow T_{\max}^{-}}||f||_{\infty}=+\infty.$
\end{lemma}
\begin{proof}
The proof is by contradiction. Let be $\limsup_{t\rightarrow
T_{\max}^{-}}||f||_{\infty}=\bar{M}<\infty$. This means that
$\forall\, \varepsilon >0$ there exists $\delta_{\varepsilon}$
such that $\forall\, \bar{t}\in
(T_{\max}-\delta_{\varepsilon},T_{\max})$ we have
$||f(\bar{t})||_{\infty}\leq \bar{M}+\varepsilon$. Now we fix
$\bar{t}>T_{\max}-\min(\delta _{\varepsilon},
T(\bar{M}+\varepsilon))$ and consider the problem
(\ref{CauchyNdim}) with initial data
$(t_0,f_0({\bf{p}}))=(\bar{t},f(\bar{t}))$. Using Lemma
\ref{Localexistence}, we have that there exists a solution $f\in
C([0,\bar{t}+T(\bar{M}+\varepsilon)]\times \mathcal{T}_{N-1})$ and
this is in contradiction with the definition of $T_{\max}$ because
$\bar{t}+T(\bar{M}+\varepsilon)>T_{\max}$.
\end{proof}
\begin{lemma}\label{estimate}
The solution $f$ of the Cauchy problem (\ref{CauchyNdim}) verifies the following a priori estimate
\begin{equation}\label{eq:festimate}
||f(t,{\bf{p}})||_{L^{\infty}}\leq \max_{{\bf{p}}}(f_0({\bf{p}}))\,e^{\mathcal{B}t},
\end{equation}
with $\mathcal{B}:=\sum_{i=1}^{N-1}\left(|\upsilon_i|+\sum_{j=1}^{N-1}|\vartheta_{i,j}|\right).$
\end{lemma}
\begin{proof}
Since $0\leq M_{{\bf{e}}_i}(f)\leq 1$, for all $i$
$$\partial_t f \leq f\,\left(\sum_{i=1}^{N-1}\left(|\upsilon_i|+\sum_{j=1}^{N-1}|\vartheta_{i,j}|\right)\right).$$\\
The proof follows easily using the Gronwall inequality.
\end{proof}
Lemma \ref{limiteinfinito} and Lemma \ref{estimate} provide the following Theorem:
\begin{theorem}[Global existence]\label{globalexistence}
There exists a unique global solution $f\in C([0,+\infty)\times \mathcal{T}_{N-1})$ to problem (\ref{CauchyNdim}).
\end{theorem}

Now we present a simple property of the moments that we will use
later in the paper to study the asymptotic behavior of the
solutions for $2\times 2$ games.
\begin{lemma}\label{momenti}
If $f\in C(\mathcal{T}_{N-1})$ then $0<M_{{\bf{k}}}(f)<1$, for all ${\bf{k}}\in \R^{N-1}$.
\end{lemma}
\begin{proof}
Let $S$ be an open set, such that
\begin{equation*}
S\subset \mathcal{T}_{N-1}\quad \mbox{with}\quad  f({\bf{p}})>0\,\,\, \forall\, {\bf{p}}\in S.
\end{equation*}
The set $S$ is not empty because $f\geq 0$ and its integral over
$\mathcal{T}_{N-1}$ is equal to $1$. We have
$f({\bf{p}})>{\bf{p}}^{{\bf{k}}}\,f({\bf{p}})>0\,\,\forall\,
{\bf{p}}\in S$, and so
\begin{equation}\label{disug}
M_{{\bf{k}}}\geq \int_S{\bf{p}}^{{\bf{k}}}f({\bf{p}})\,d{\bf{p}} >0.
\end{equation}
We have also that
\begin{equation}\label{disug0}
0<\int_S {\bf{p}}^{{\bf{k}}}\,f({\bf{p}})\,d{\bf{p}}<\int_S f({\bf{p}})\,d{\bf{p}},
\end{equation}
Since $\int_{\mathcal{T}_{N-1}\setminus S}f({\bf{p}})\,d{\bf{p}}\geq 0$ and (\ref{disug0}) holds, we have
$$1=\int_{\mathcal{T}_{N-1}} f({\bf{p}})\,d{\bf{p}}=\int_{\mathcal{T}_{N-1}\setminus S}f({\bf{p}})\,d{\bf{p}} +\int_S f({\bf{p}})\,d{\bf{p}}>M_{{\bf{k}}}(f),\quad \forall\, {\bf{k}}.$$
\end{proof}

\section{Two strategies games}
Assume there are two different strategies, whose interplay is ruled by the payoff matrix:
\[ \acal = \left(\begin{matrix} a & b \cr c & d\end{matrix}\right).\]
In this case the simplex $\mathcal{T}_1$ is just the interval
$[0,1]$ and so we have a population where individuals are going to
play the first strategy with probability $p\in [0,1]$ and the
second strategy with probability $1-p$. The payoff
(\ref{eq:payoff}) is given by
\begin{equation}\label{mix_payoff}
\begin{array}{ll}
A({\bf{p}},{\bf{p^*}}):=&\left(\begin{array}{c} p \\ 1-p \end{array}\right) \left(\begin{matrix} a & b \cr c & d\end{matrix}\right)
\left(\begin{array}{c} p^* \\ 1-p^* \end{array}\right)\\ \\
&=(a+d-b-c)pp^*+(b-d) p+(c-d)p^*+d\\ \\&=
\alpha pp^*+\beta p +\gamma p^*+\delta, \end{array}\end{equation}
with
\begin{equation}\label{eq:alphabeta}
\alpha:=a+d-b-c, \quad \beta:=b-d,
\quad \gamma:=c-d, \quad \delta:=d.
\end{equation}
The one dimensional Cauchy problem (\ref{CauchyNdim}) reads
\begin{equation}\label{eq:model}
\begin{cases}
\partial_t f(p)=f(p)\left[(\alpha M_1(f)+\beta)(p-M_1(f))\right] \quad t\geq 0,\,p \in[0,1], \\
f(0,p)=f_0(p), \qquad p\in [0,1],
\end{cases}
\end{equation}
with $f_0(p)\geq 0$ and $\int_0^1 f_0(p)dp=1$.\\
\subsection{Asymptotic behavior of the solutions}
We want to understand what happens asymptotically. We start with a result on the curve of change of sign for $\partial_t\,f$.
\begin{proposition}\label{curvesign}
If $f\in C([0,1])$ then for all $t \geq 0$ there exists $\bar{p}=\bar{p}(t)\in (0,1)$ such that $M_1(f(t))=\bar{p}(t)$, namely
\begin{eqnarray} \label{eq:derivataf}
&\partial_t f(t,\bar{p}(t))&=0, \\
&\sgn(\partial_t f(t,p))&=-\sgn(\alpha M_1(f(t,p))+\beta) \qquad \forall\, p<\bar{p}(t), \\
&\sgn(\partial_t f(t,p))&=\sgn(\alpha M_1(f(t,p))+\beta) \qquad \forall\, p>\bar{p}(t).
\end{eqnarray}
\end{proposition}
The proof of Proposition \ref{curvesign} is easily obtained by Lemma \ref{momenti}.\\

Let us write the equation for the first moment $M_1(f)$:
\begin{equation}\label{eq:M1}
M'_1(t)=(\alpha M_1(f)+\beta)(M_2(f)-M_1^2(f)).
\end{equation}
The Jensen inequality gets
\begin{equation} \label{eq:M1dis}
M_1^2(f)=\left( \int_0^1 p\, f\, dp\right)^2\leq \int_0^1 p^2 \, f \,dp=M_2,
\end{equation}
and so
\begin{equation}\label{eq:M1sgn}
\sgn(M'_1(t))=\sgn(\alpha M_1(f)+\beta).
\end{equation}
There are four different possible cases:
\begin{description}
\item[Case a] $-\dfrac{\beta}{\alpha} \notin (0,1)$. \figurename~\ref{fig:ABCD} shows that it is possible if and only if $(\alpha, \beta) \in A\,\cup
B$.\\
\begin{figure}[htp]
\center \epsfig{file=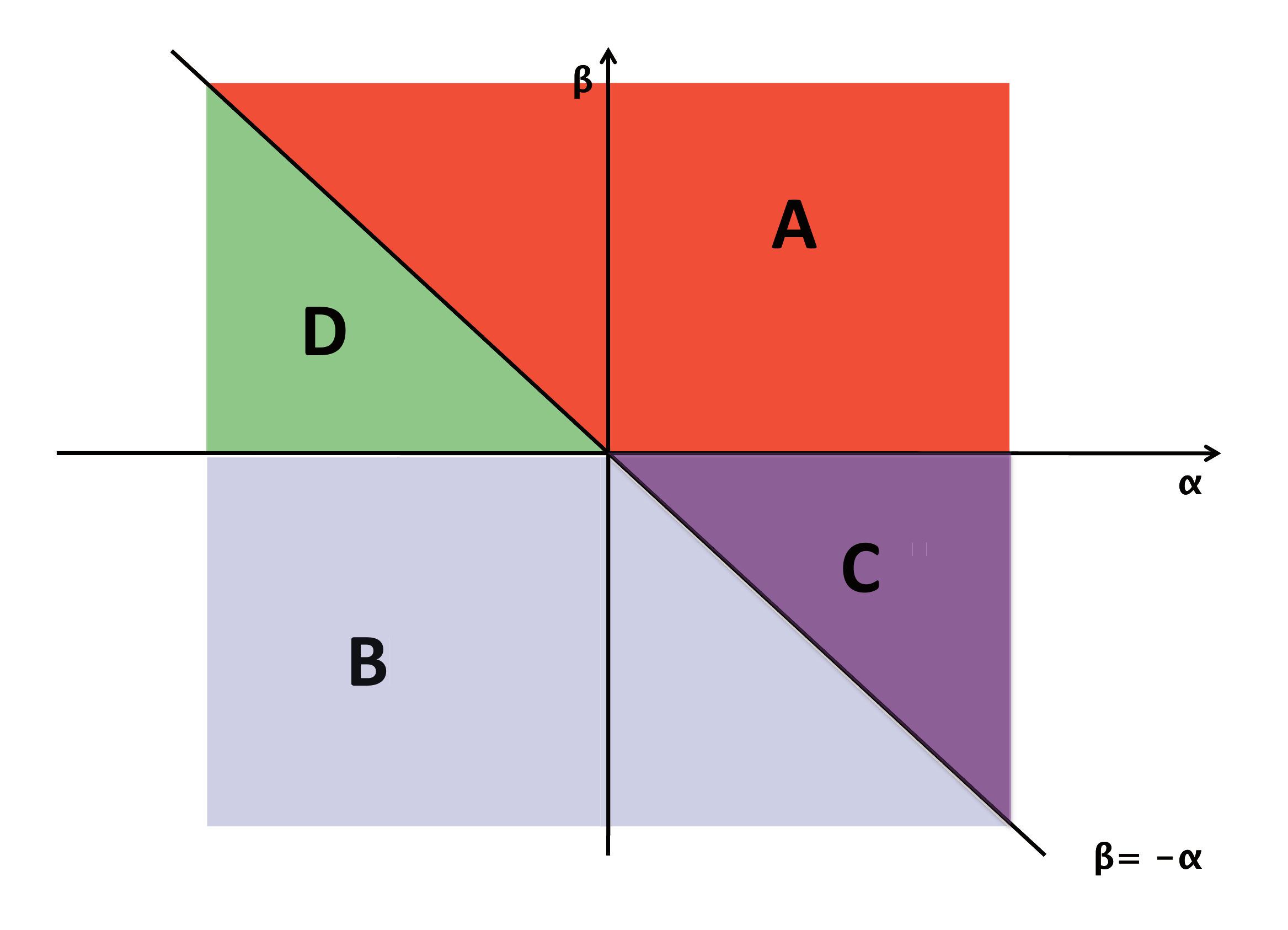,scale=0.4} \caption{The quantity
$-\dfrac{\beta}{\alpha}\notin (0,1) \Longleftrightarrow
(\alpha,\beta) \in A\cup B.$ In $A$ we have
$\beta>\min(0,-\alpha)$; in $B$ we have $\beta < \min(0,-\alpha)$.
The quantity $-\dfrac{\beta}{\alpha}\in (0,1) \Longleftrightarrow
(\alpha,\beta) \in C\cup D.$ In $C$ we have $\alpha >0$, $-\alpha
<\beta <0$; in $D$ we have $\alpha<0$, $0<\beta <-\alpha$.}
\label{fig:ABCD}
\end{figure}
Let us describe in detail the different situations:\\
\begin{description}
\item[$(\alpha,\beta)\in A$] this means that $\beta>\min(0,-\alpha)$. If $\alpha \geq 0$ then $0<\beta\leq \alpha M_1(f)+\beta \leq \alpha+\beta$; if $\alpha < 0$ then $0<\alpha+\beta\leq \alpha M_1(f)+\beta\leq \beta$. In any case we have $\alpha M_1(f)+\beta\geq 0$ and so $M'_1(f)=(\alpha M_1(f)+\beta)(M_2(f)-M_1^2(f))\geq 0.$ As shown in \figurename~\ref{fig:AB} (on the left), $M_1(f)$ is increasing in time and is limited on the right by the curve $\tilde{M}(t)\longrightarrow 1$ with $\tilde{M}'(t)=(\alpha \tilde{M}+\beta)\tilde{M}(1-\tilde{M}).$ \\
\item[$(\alpha,\beta)\in B$] this means that $\beta<\min(0,-\alpha)$. If $\alpha \geq 0$ then $\beta\leq \alpha M_1(f)+\beta \leq \alpha+\beta<0$; if $\alpha < 0$ then $\alpha+\beta\leq \alpha M_1(f)+\beta\leq \beta<0$. In any case we have $\alpha M_1(f)+\beta\leq 0$ and so $M'_1(f)=(\alpha M_1(f)+\beta)(M_2(f)-M_1^2(f))\leq 0.$ As shown in \figurename~\ref{fig:AB} (on the right), $M_1(f)$ is decreasing in time and is limited on the left by the curve $\tilde{M}(t)\longrightarrow 0$ with $\tilde{M}'(t)=(\alpha \tilde{M}+\beta)\tilde{M}(1-\tilde{M}).$
\begin{figure}
\subfigure {\epsfig{file=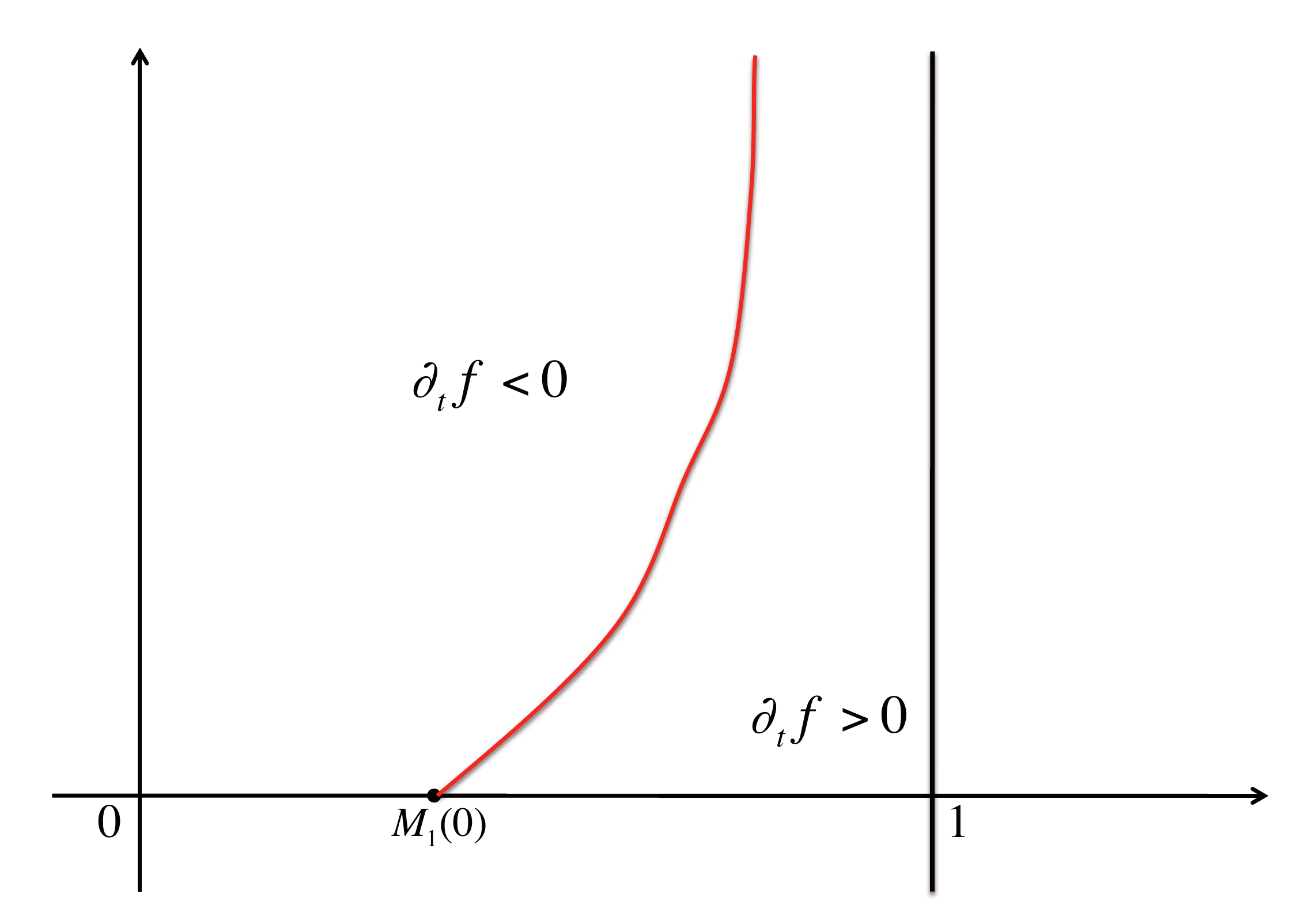,height=3.2cm,width=5.9cm}}
\hspace{5mm} \subfigure
{\epsfig{file=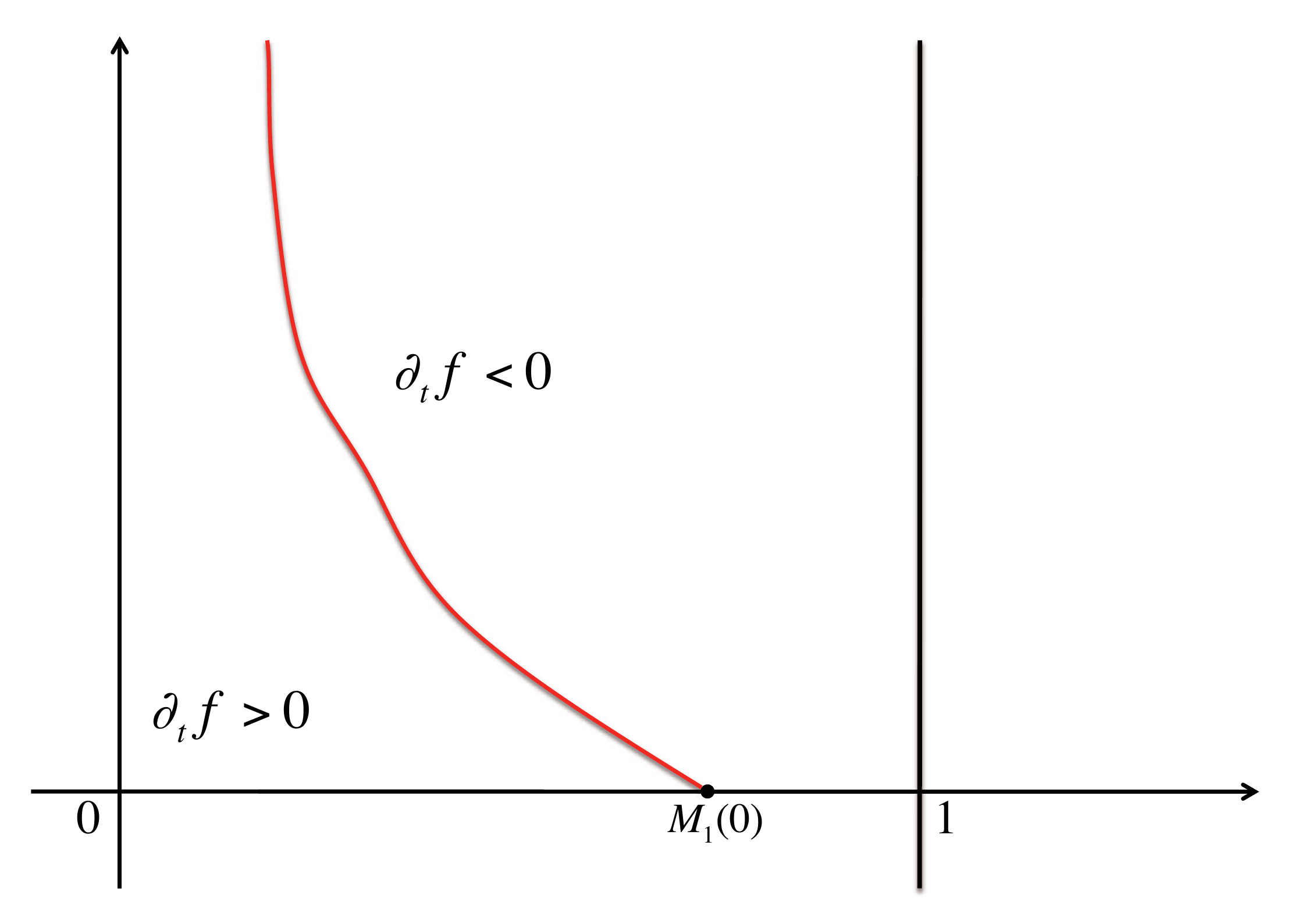,height=3.2cm,width=5.9cm}} \caption{On
the left: evolution of $M_1(t)$ in the case $(\alpha, \beta)\in
A$. On the right: evolution of $M_1(t)$ in the case $(\alpha,
\beta)\in B$.}\label{fig:AB}
\end{figure}
\end{description}
\vspace{0.2cm}
\item[Case b] $-\dfrac{\beta}{\alpha} \in (0,1)$. \figurename~\ref{fig:ABCD} shows that it is possible if and only if $(\alpha, \beta) \in C\,\cup D$.\\
\begin{description}
\item[$(\alpha,\beta)\in C$] \figurename~\ref{fig:CD} (on the left) shows two different situations: $$M_1(0)>-\dfrac{\beta}{\alpha}\quad \Longrightarrow \quad M'_1(f)>0;$$
$$M_1(0)<-\dfrac{\beta}{\alpha}\quad \Longrightarrow \quad M'_1(f)<0.$$
By contrast with the previous case, the behavior changes according to the value of the first moment $M_1$ at initial time $t=0$. If $M_1(0)>-\frac{\beta}{\alpha}$ then $M_1(t)$ increases in time and is limited on the left by the curve $\tilde{M}(t)\longrightarrow 1$ with $\tilde{M}'(t)=(\alpha \tilde{M}+\beta)\tilde{M}(1-\tilde{M}).$ Conversely, if $M_1(0)<-\frac{\beta}{\alpha}$ then $M_1(t)$ decreases in time and is limited on the right by the curve $\tilde{M}(t)\longrightarrow 0$ with $\tilde{M}'(t)=(\alpha \tilde{M}+\beta)\tilde{M}(1-\tilde{M}).$\\

\item[$(\alpha,\beta)\in D$] \figurename~\ref{fig:CD} (on the right) shows the  two situations: $$M_1(0)>-\dfrac{\beta}{\alpha}\quad \Longrightarrow \quad M'_1(f)<0;$$
$$M_1(0)<-\dfrac{\beta}{\alpha}\quad \Longrightarrow \quad M'_1(f)>0.$$
Also in this case, the behavior depends on the value of $M_1(0)$:
if $M_1(0)>-\frac{\beta}{\alpha}$ then $M_1(t)$ decreases in time
away from the value $-\frac{\beta}{\alpha}$; if
$M_1(0)<-\frac{\beta}{\alpha}$ then $M_1(t)$ increases in time
toward the value $-\frac{\beta}{\alpha}$. In any cases, the value
$-\frac{\beta}{\alpha}$ is the one that dominates in time.
\begin{figure}[htp]
\subfigure
{\epsfig{file=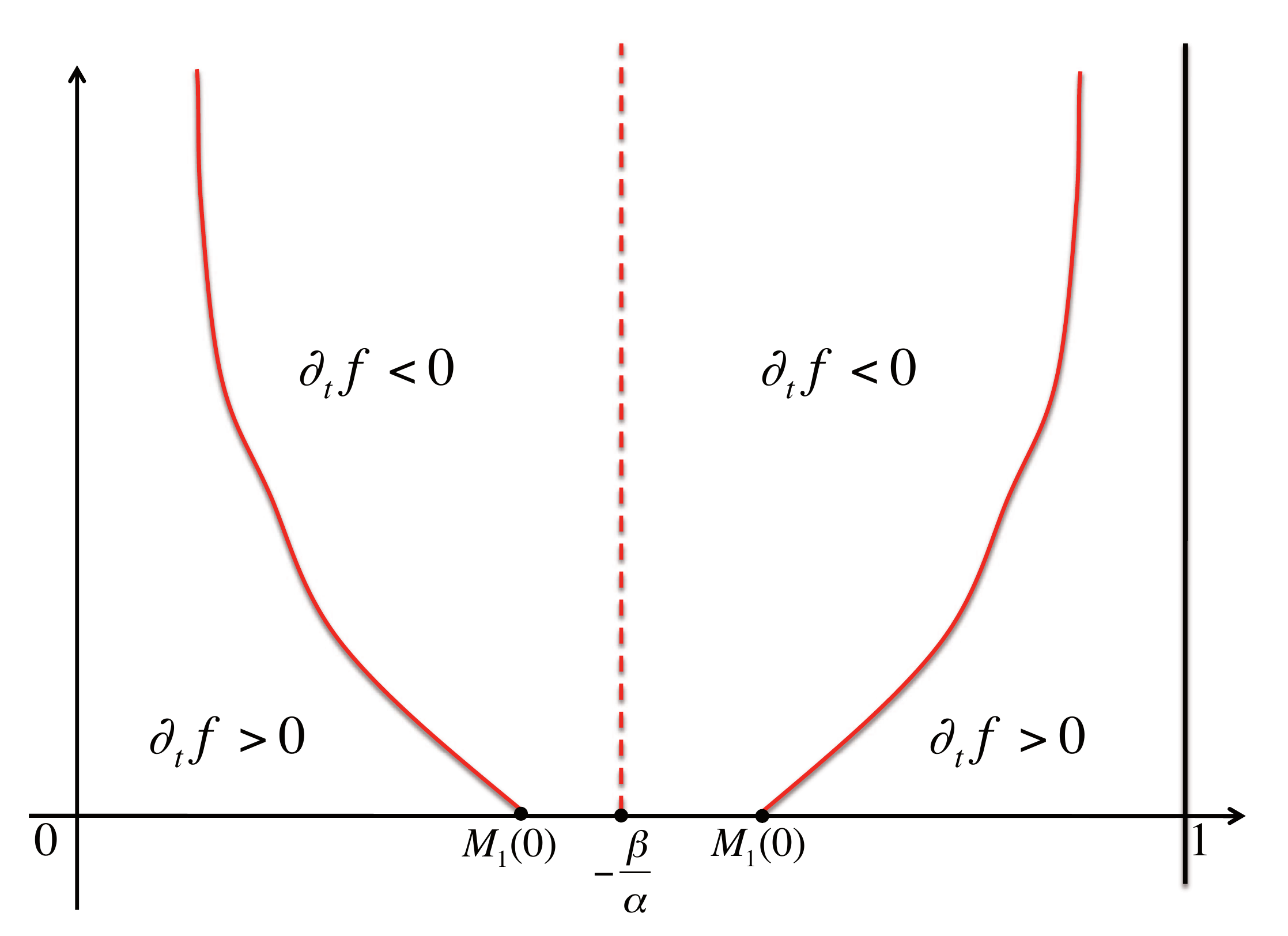,height=3.2cm,width=5.9cm}}
\hspace{5mm}
\subfigure
{\epsfig{file=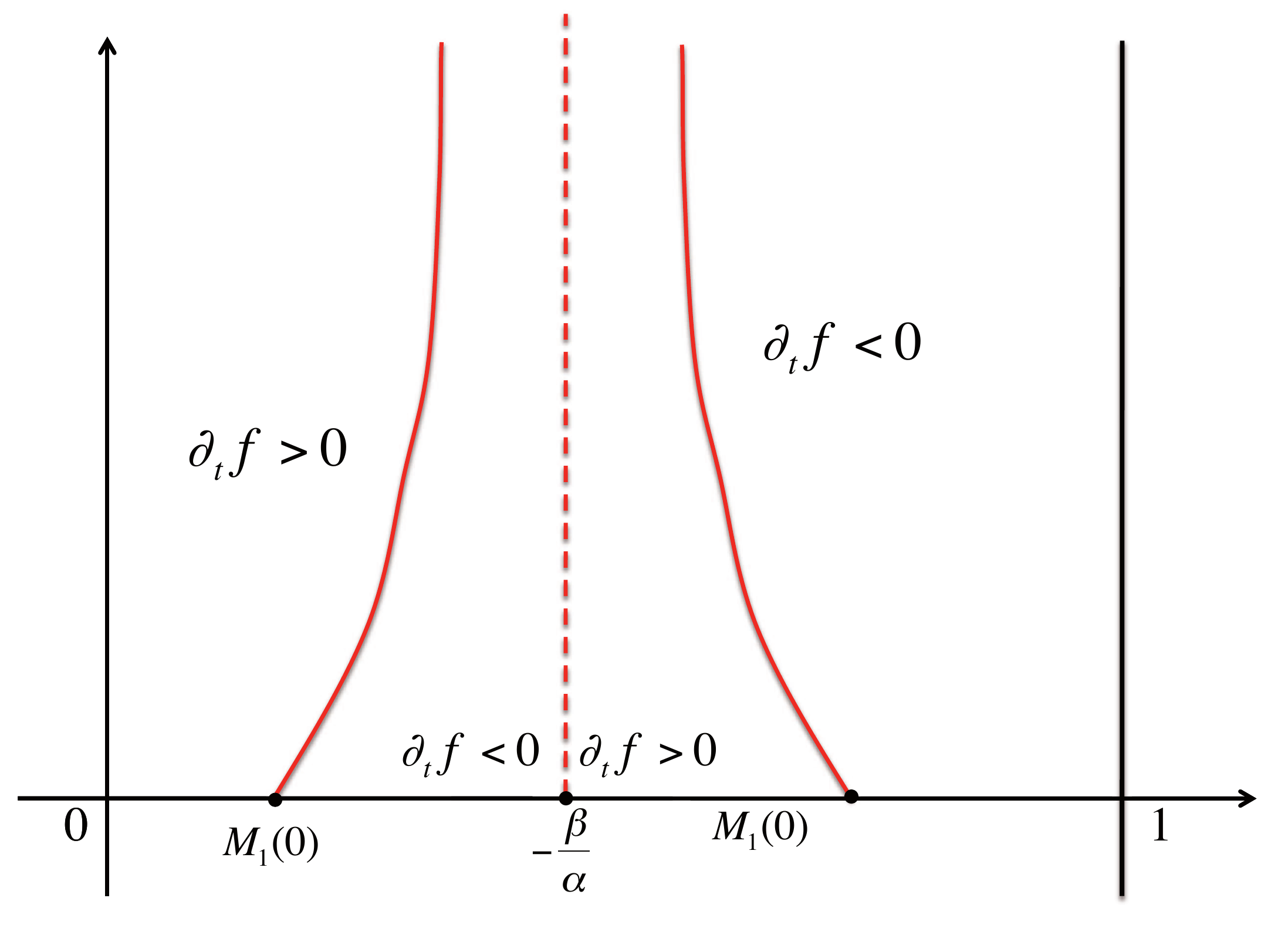,height=3.2cm,width=5.9cm}}
\caption{On the left: evolution of $M_1(t)$ in the case $(\alpha, \beta)\in C$. On the right: evolution of $M_1(t)$ in the case $(\alpha, \beta)\in D$.}\label{fig:CD}
\end{figure}
\end{description}
\end{description}
From the behavior described is easy to understand what happens in
the population. If  we are in region $A$, there is dominance of
the first of the two pure strategies that describe the game,
because the dynamic encourages the state $p=1$ which corresponds
to the first pure strategy. This means that in the population,
those who adopt the first pure strategy survive, the others do
not. In $B$ we have the opposite situation: there is the dominance
of the second pure strategy and so those who adopt the first pure
strategy or any other mixed strategy, do not survive. The third
region $C$ is such that there is not a mixed strategy that
dominates, but a priori we can not say which of the two pure
strategies dominates, it all depends on the value of $M_1(0)$. If
$M_1(0)>-\frac{\beta}{\alpha}$ then there is the dominance of the
first pure strategy, if $M_1(0)<-\frac{\beta}{\alpha}$ then there
is the dominance of the second pure strategy. In $D$ we have a
different situation than in the previous cases: here there is
coexistence between the two pure strategies and so between the two
populations.
\subsection{Stationary solutions} \label{stasol}
From the study of the asymptotic behavior, we expect that for
$t\rightarrow \infty$ the solution of the model (\ref{eq:model})
tends to a stationary solution. We can find two classes of
stationary solutions:
\begin{description}
\item[Type I]
If we are in case b, namely $-\frac{\beta}{\alpha}\in (0,1)$, then
a stationary solution is given by every density function
$\overline f(p)$ such that
\begin{equation}\label{M1case1}
M_1(\overline f)=-\frac{\beta}{\alpha}.
\end{equation}
\item[Type II]
If we consider more general solutions, so that the probability
distributions are no more absolutely continuous with respect to
the Lebesgue measure, we can say that another class of stationary
solutions is given by concentrated Dirac masses, i.e.:
\[  f_{\overline p}(p)=\delta (p=\overline p).\]
In the following we are going to deal with these generalized solutions in a quite informal way.
More rigorous arguments will be given in a future paper.\\

Here we want just remark that formally, since $M_1(f_{\overline p})=\overline p$, we have
\[
f_{\overline p}(p)\left[(\alpha M_1(f_{\overline p})+\beta)(p-M_1(f_{\overline p}))\right]=\delta (p=\overline p)(\alpha\overline p+\beta)(p-\overline p)=0.\]
\end{description}
\subsubsection{Linear stability of stationary solutions}
This Subsection is dedicated to the study of the linear stability of stationary solutions. Denote by
\[ Q(f) =f(p)\left[(\alpha M_1(f)+\beta)(p-M_1(f))\right]\]
the integral operator associated to the replicator equation. Let
$\tilde f$ be a generalized stationary state. We linearise the
operator around the state  $\tilde f$. So for every perturbation
$g$, with $\int_0^1 g(p) dp=0$, we have the linear operator
\[
\begin{array}{ll}
Q'(\tilde f)(g)&=\lim_{h\to0}\frac{1}{h}\left(Q(\tilde f+hg) -Q(\tilde f)\right)\\ \\
&=\left[\tilde f M_1(g)+g(\alpha M_1(\tilde f)+\beta)\right](p-M_1(\tilde f)). \end{array}\]
\begin{description}
\item[Type I] we have $-1<\dfrac{\beta}{\alpha}<0$. Using (\ref{M1case1}) we obtain that the linearized equation for a perturbation $g$ is given by
\begin{equation}\label{linear1}
\partial_t g(p)=Q'(\overline f)(g)=\overline f(p) (p-M_1(\overline f))M_1(g).
\end{equation}
If $\int_0^1 g_0(p) dp=0$, the same is true for $g$ for $t>0$. \\
\begin{proposition}\label{prop:nostable}
Assume $-\dfrac{\beta}{\alpha}\in(0,1)$. Then, there is no
continuous stationary solution  $\bar{f}(p)$  to  problem
(\ref{eq:model}) which is linearly stable.
\end{proposition}
\begin{proof}
To prove the result it is enough to establish the following equality
\begin{equation}\label{eq:stability}
M_2(\bar{f})=(M_1(\bar{f}))^2,
\end{equation}
which implies that the variance of the measure $\bar{f} dp$
vanishes and so the measure has to be a Dirac mass. We take a
continuous stationary solution of Type I, namely a positive
function $\overline f$ such that its total mass is equal to 1 and
$M_1(\overline f)=-\frac{\beta}{\alpha}$. We perturb this state by
a function $g$ of zero mass. Computing the first moment of the
perturbation $g$ yields
\[M'_1(g)=M_1(g)(M_2(\overline f)-(M_ 1(\overline f))^2).\]
Setting
$$N_1(\bar{f}):=M_2(\overline f)-(M_1(\overline f))^2=\int_0^1 p(p+\frac{\beta}{\alpha}) {\overline f(p)}dp,$$
we obtain:
\[M_1(t,g)=M_1(0,g)e^{tN_1}.\]
Moreover, setting $N_k(\overline f):=M_{k+1}(\overline
f)-M_{k}(\overline f)M_1(\overline f),$ we have
$$M_k(t,g)=M_k(0,g)+\frac{N_k(\overline f)
M_1(0,g)}{N_1(\bar{f})}(e^{tN_1(\bar{f})}-1).$$ This means that
the condition for linear stability is just:
$$N_1(\overline f)=\int_0^1 p(p+\frac{\beta}{\alpha}) \overline f(p) dp\leq 0.$$
This inequality can be verified only when the equality condition is satisfied, since we already know from (\ref{eq:M1dis}) that $N_1(\bar{f})\geq 0$.
\end{proof}
\item[Type II]
For the concentrated Dirac masses $f_{\overline p}(p)=\delta
(p=\overline p)$ we have that  the linearized equation for a
perturbation $g$ is given by
\begin{equation}\label{linear2}
\partial_t g(p)=Q'(f_{\overline p})(g)= g(p)(\alpha \overline p+\beta)(p-\overline p) .
\end{equation}
\begin{proposition}
The Dirac mass solutions are linear stable if, on the support of $g(p)$, we have:
\[(\alpha \overline p+\beta)(p-\overline p) \leq 0,\qquad \forall\, p\in[0,1].\]
For general perturbations, i.e.: with $supp\  g\equiv[0,1]$ we have three cases:
\begin{enumerate}
\item the Dirac mass concentrated in $\overline p=0$ is stable if $\beta<0$, which means in the original constants, $b<d$.
\item the Dirac mass concentrated in $\overline p=1$ is stable if $\alpha +\beta>0$, which means in the original constants, $a>c$.
\item if $-1<\frac{\beta}{\alpha}<0$, then the Dirac mass concentrated in  $\overline p=-\frac{\beta}{\alpha}$ is stable.
\end{enumerate}
\end{proposition}
\end{description}

\section{Numerical approximation of the 1D model}
We want to perform some numerical simulations with model (\ref{eq:model}). We consider some nodes
\begin{equation*}
p_i\in [0,1] \qquad i=0,\dots, I,\quad (p_0=0,\,p_I=1),
\end{equation*}
and, using the trapezoidal rule, we obtain the
following quadrature formula for the first moment $M_1(f)$ (for all
times $t\in [0,T]$)
\begin{equation}\label{eq:discretize}
M_1(f)=\int_0^1 pf(t,p)dp\approx \tilde{M_1}(f):=\displaystyle \sum_{i=0}^{I-1} \dfrac{p_{i+1}-p_i}{2}\,\left[p_{i+1}f(t,p_{i+1})+p_if(t,p_i)\right].
\end{equation}
Clearly the above discretization is such that conservation of mass
holds
\begin{eqnarray*}
&&\frac{d}{dt} \sum_{i=0}^{I-1} \dfrac{p_{i+1}-p_i}{2}\,\left[f(t,p_{i+1})+f(t,p_i)\right] =\\ &&(\alpha \tilde{M_1}(f)+\beta)\sum_{i=0}^{I-1} \dfrac{p_{i+1}-p_i}{2}\left[(p_{i+1}-\tilde{M_1}(f))f(t,p_{i+1})+(p_i-\tilde{M_1}(f))f(t,p_i)\right]=0,
\end{eqnarray*}
provided that initially
\[
\sum_{i=0}^{I-1} \dfrac{p_{i+1}-p_i}{2}\,\left[f_0(p_{i+1})+f_0(p_i)\right]=1.
\]
In the case of equally spaced points, $\Delta p= p_{i+1}-p_i$, the above property implies that
\[
f(t,p_i)\leq \frac1{\Delta p},\quad \forall\, t>0,\quad i=0,\ldots, I,
\]
and thus the numerical solution is well-defined even when we approach a Dirac delta at the
continuous level. More precisely both possible steady states are preserved by the numerical method,
namely $\tilde{M_1}(\bar{f})=-{\beta}/{\alpha}$ and the discrete Dirac delta defined as
\[
f_{\bar p_i}(p_j)=\left\{
                         \begin{array}{ll}
                           0, & {i\neq j} \\
                           \displaystyle\frac1{\Delta p}, & {i=j,\,i=1,\ldots,I-1}\\[+.3cm]
                           \displaystyle\frac2{\Delta p}, & {i=j,\,i=0,I}.
                         \end{array}
                       \right.
\]
For the time discretization we simply use a fourth order Runge-Kutta
method with constant time stepping $\Delta t$ on the interval $[0,T]$. Nonnegativity of
the numerical solution is achieved taking
\[
\Delta t \leq |\alpha \tilde{M_1}(f^n)+\beta|^{-1},
\]
where $f^n(p_i)=f(t^n,p_i)$, $t^n=n\Delta t$.

\subsection{Numerical tests: Prisoner's Dilemma game}
One interesting example of a game is given by the so-called
Prisoner's dilemma game in which there are two players and two
possible strategies. The players have two options, cooperate or
defect. The payoff matrix is the following
\begin{eqnarray*}
C \qquad D \quad\,\,\, &\\
\\[-.3cm]
\mathcal{A}=\left(\begin{array}{ccc}
R & & S \\ T & & P
\end{array}\right).&\begin{array}{c} C \\[+.0cm] D \end{array}
\end{eqnarray*}
If both players cooperate both obtain $R$ fitness units (reward
payoff); if both defect, each receives $P$ (punishment payoff); if
one player cooperates and the other defects, the cooperator gets
$S$ (sucker's payoff)  while the defector gets $T$ (temptation
payoff). The payoff values are ranked $T>R>P>S$ and $2R>T+S$. From
the game theory we know that cooperators are always dominated by
defectors. One of the main problems has been about the possibility
of success for cooperation, which is impossible in the pure
strategies models: the replicator dynamics of prisoner's dilemma,
\cite{HofbauerSigmund}, shows that cooperators are extinguished.

For the numerical tests we fix the following normalized payoff matrix:
\begin{equation}\label{PDmatrix}
\acal = \left(\begin{matrix} 1 & 0 \cr b & \varepsilon\end{matrix}\right),
\end{equation}
with $b=1.1$ and $\varepsilon=0.001$. In this case we have
$\alpha=1-b+\varepsilon<0$ and $\beta=-\varepsilon <0$ and so
$\frac{\beta}{\alpha}>0$. This means that stationary solutions are
expected to be given by concentrated Dirac masses (see Section
\ref{stasol}). For general perturbation we have that $\bar{p} = 0$
is linearly stable.

\subsubsection{Test n.1}
We consider the initial datum
\begin{equation}
f_0(p)=1 \qquad \forall\, p\in [0,1]. \label{eq:dato1}
\end{equation}
\begin{figure}[htp]
\center \epsfig{file=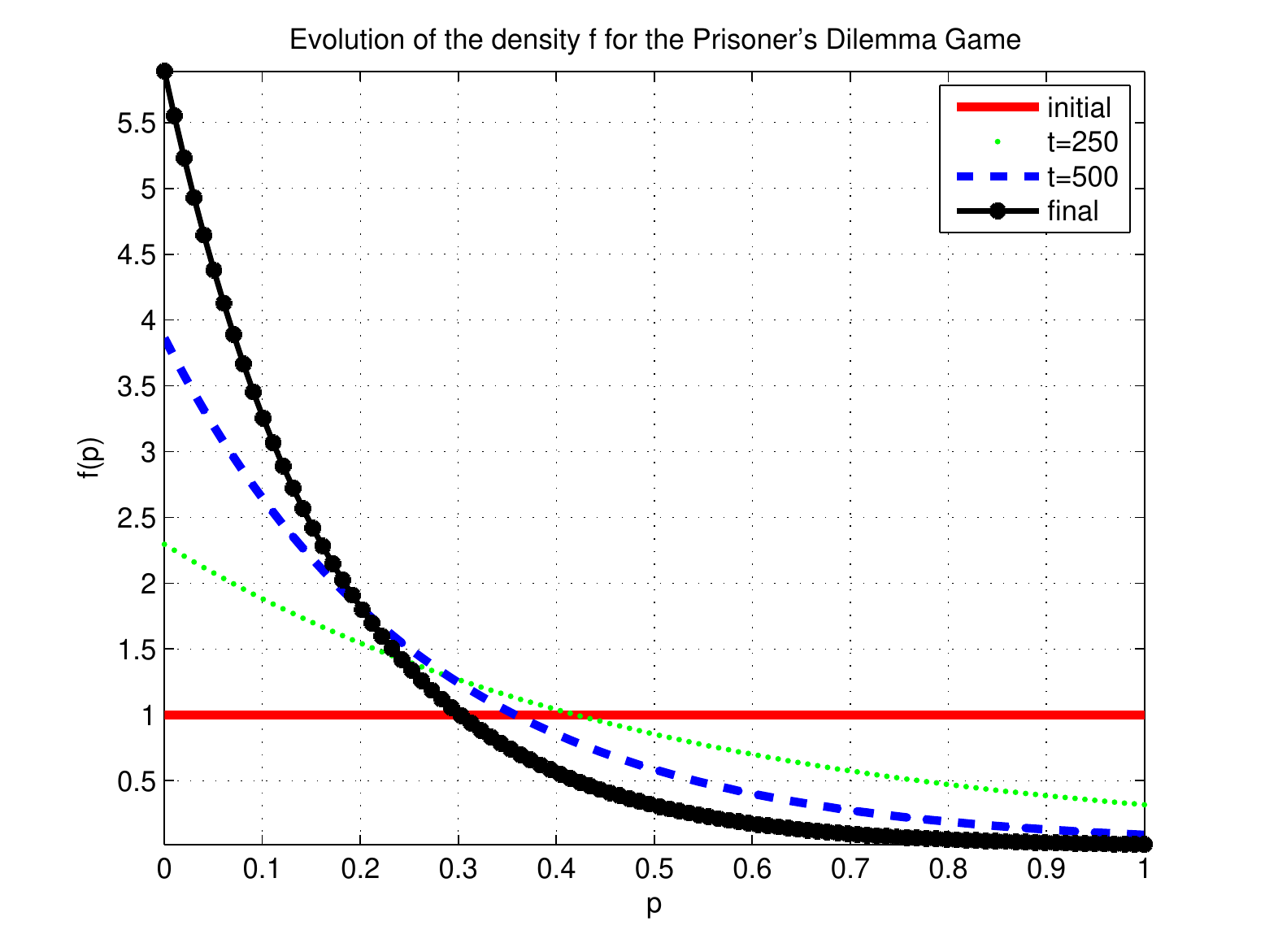,height=6cm,width=11cm}
\caption{Prisoner's Dilemma Game, test n.1: $b=1.1$,
$\varepsilon=0.001$. Plot of the evolution over time of $f(t,p)$
for the Cauchy problem (\ref{eq:model}) with $f_0(p)=1$ for
$T=1000$.} \label{fig:test1}
\end{figure}
\figurename~\ref{fig:test1} shows that the density $f$ tends to concentrate at the point $p=0$, according to what we expected. \\\\
We have studied, numerically, the $L^{\infty}$-norm of the
solution $f$. Using the a priori estimate (\ref{eq:festimate}), we
know that $$||f||_{\infty}\leq
||f_0(p)||_{\infty}e^{(|\alpha|+|\beta|)t}=e^{(|\alpha|+|\beta|)t},$$
and this yields
\begin{equation}\label{stimaplot}
H(t):=Log(||f||_{\infty})\leq (|\alpha|+|\beta|)e^{Log(t)}=:E(t), \qquad \forall\, t.
\end{equation}
\figurename~\ref{fig:test1b} shows that inequality (\ref{stimaplot}) is respected.
\begin{figure}[htp]
\center \epsfig{file=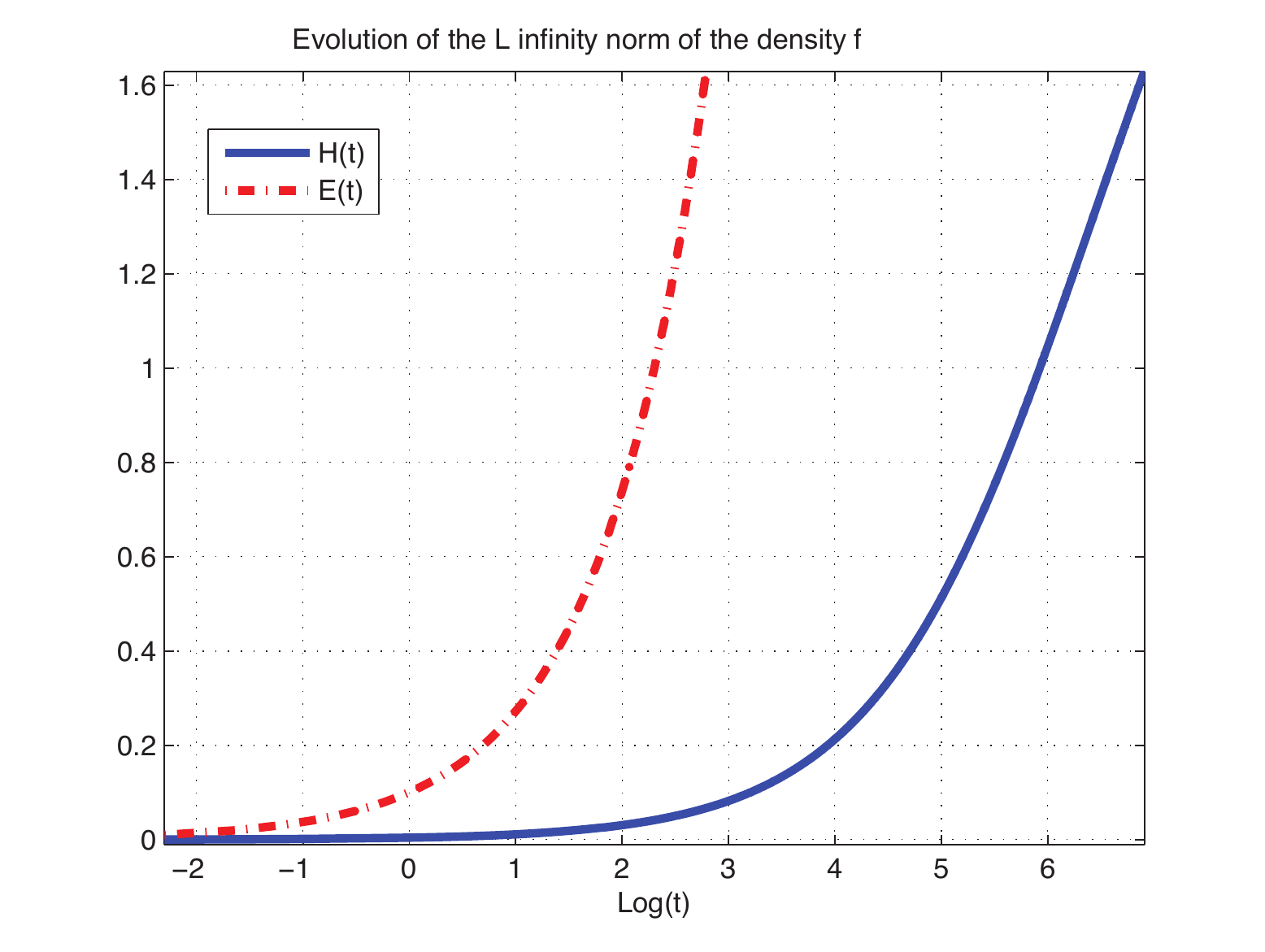,height=6cm,width=11cm}
\caption{Prisoner's Dilemma Game, test n.1: we have $Log(t)$ on
the x-axes and $H(t)$ and $E(t)$ on the y-axes for the Cauchy
problem (\ref{eq:model}) with $f_0(p)=1$ for $T=1000$.}
\label{fig:test1b}
\end{figure}
\begin{figure}[htp]
\center \epsfig{file=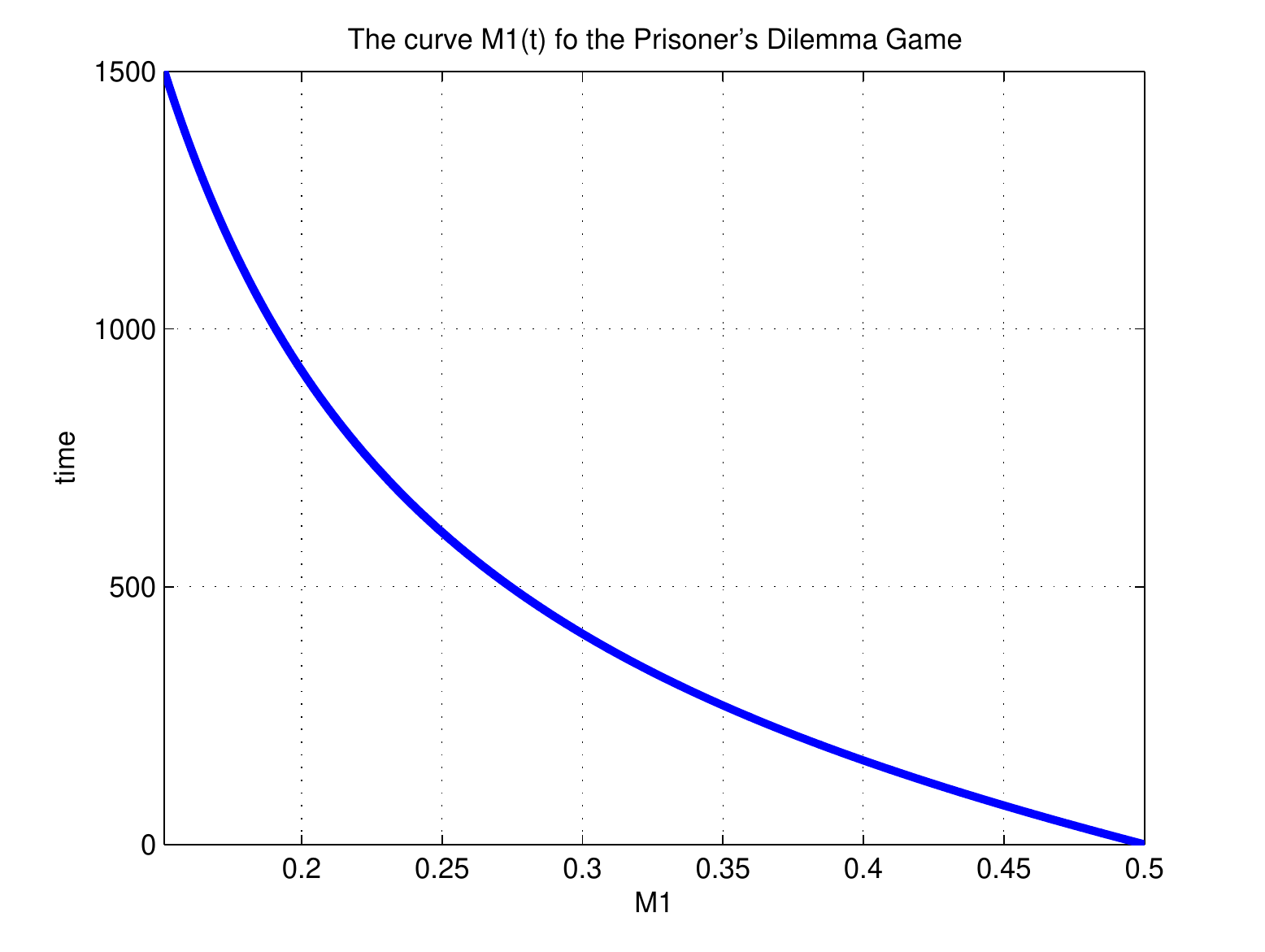,height=6cm,width=11cm}
\caption{Prisoner's Dilemma Game, test n.1: $b=1.1$,
$\varepsilon=0.001$. Plot of the evolution of $M_1(t)$ vs. time
$t$ for the Cauchy problem (\ref{eq:model}) with $f_0(p)=1$ for
$T=1500$.} \label{fig:test1c}
\end{figure}
In the prisoner's dilemma game we have that $(\alpha, \beta)\in B$
(see \figurename~\ref{fig:ABCD}) and we know, from game theory,
that the defectors' pure strategy dominates the cooperators' pure
strategy.
The evolution in time of $M_1(f)$ (\figurename~\ref{fig:test1c}) is as expected (see  \figurename~\ref{fig:AB} (on the right)). \\

We consider now a quadratic initial datum for the model
(\ref{eq:model}). We have plotted the numerical results in
\figurename~\ref{fig:test1e}. As in the previous case with
$f_0(p)=1$, we see that the density $f$ tends to concentrated at
the point $p=0$ that corresponds to the defectors' strategy.

\begin{figure}[htp]
\center \epsfig{file=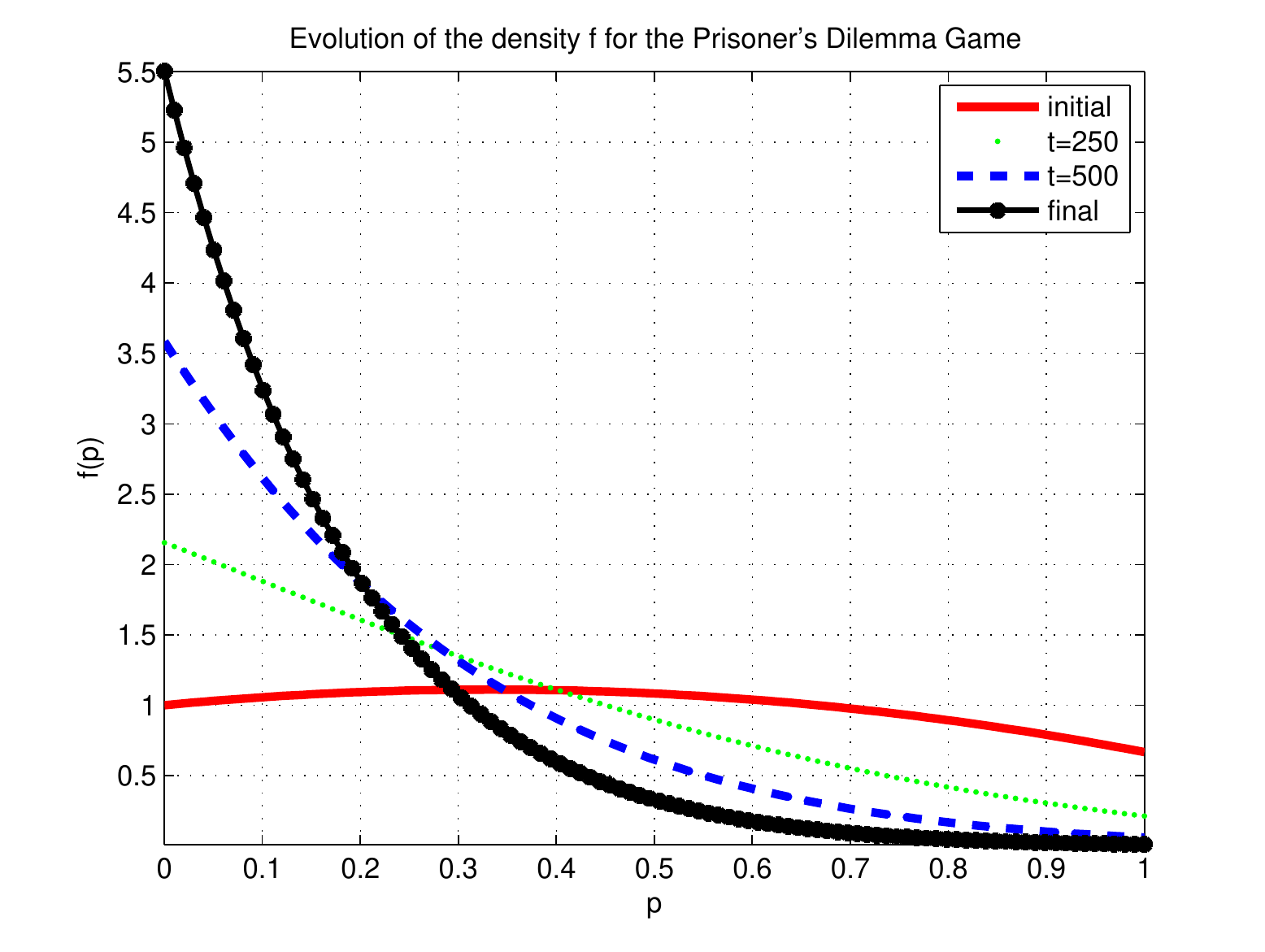,height=6cm,width=11cm}
\caption{Prisoner's Dilemma Game, test n.1: $b=1.1$,
$\varepsilon=0.001$. Plot of the evolution over time of $f(t,p)$
for the Cauchy problem related to the model (\ref{eq:model}) with
initial datum $f_0(p)=-p^2+\frac{2}{3}p+1 \quad \forall\, p\in
[0,1]$ for $T=1000$.} \label{fig:test1e}
\end{figure}
\subsubsection{Test n.2}
Now we want to consider an initial datum $f_0(p)$ with compact
support in $[p_1,p_2]\subset [0,1]$, with $p_1<p_2$. If we define
\begin{equation}
q:=\dfrac{p-p_1}{p_2-p_1}
\end{equation}
we have that $f_0(q)$ has compact support in $[0,1]$. W.r.t. $q$ the average payoff (\ref{mix_payoff}) has the following form
\begin{equation}
A(q,q^{*})=\bar{\alpha}qq^{*}+\bar{\beta}q+\bar{\gamma}q^{*}+\bar{\delta}, \label{mix_payoff_q}
\end{equation}
with
\begin{equation} \label{eq:const_q}
\bar{\alpha}:=\alpha (p_2-p_1)^2, \quad \bar{\beta}:=\alpha (p_2-p_1)p_1+\beta (p_2-p_1),
\end{equation}
\begin{equation}\label{eq:const_q1}
\bar{\gamma}:=\alpha (p_2-p_1)p_1+\gamma (p_2-p_1), \quad \bar{\delta}:=\alpha p_1^2+(\gamma-\beta-2\delta)p_1+\delta.
\end{equation}
The quantity
$$\dfrac{\bar{\beta}}{\bar{\alpha}}=\dfrac{1}{p_2-p_1}\left(p_1+\dfrac{\beta}{\alpha}\right)$$
is positive if $\frac{\beta}{\alpha}>0,$ as in the Prisoner's
Dilemma game. This means that the point $\bar{p}=p_1$
(corresponding to the point $\bar{q}=0$) is stable, as shown in
\figurename~\ref{fig:test2} that is related to the Cauchy problem
(\ref{eq:model}) with the following initial datum:
\begin{equation}\label{eq:dato4}
f_0(p)=\begin{cases}
2 \qquad p\in \left[\frac{1}{4},\frac{1}{2}\right]\cup \left[\frac{3}{4},1\right] \\
0 \qquad \qquad \mbox{elsewhere.}
\end{cases}
\end{equation}
\begin{figure}[htp]
\center \epsfig{file=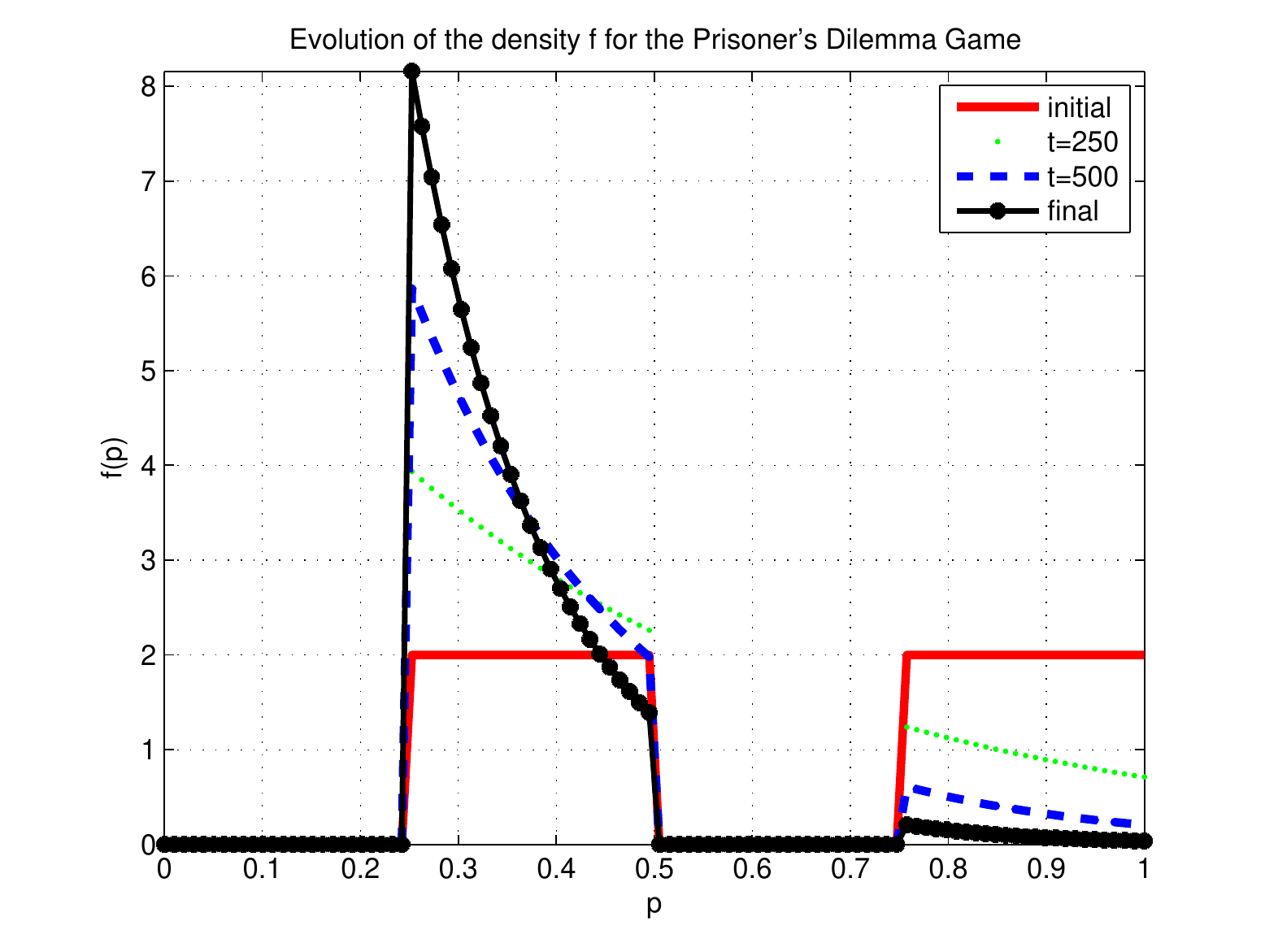,height=6cm,width=11cm}
\caption{Prisoner's Dilemma Game, test n.2: $b=1.1$,
$\varepsilon=0.001$. Plot of the evolution over time of $f(t,p)$
for the Cauchy problem (\ref{eq:model}) with initial datum
(\ref{eq:dato4}) for $T=1000$.} \label{fig:test2}
\end{figure}

\subsection{Numerical tests: Hawk or Dove Game}
Another example of a game is given by the so-called Hawk or Dove
game in which there are two pure strategies: hawks (H) and doves
(D). While hawks escalate fights, doves retreat when the opponent
escalates. The benefit of winning the fight is $b$. The cost of
injury is $c$. If two hawks meet, then the expected payoff for
each of them is $\frac{b-c}{2}$. The fight will escalate. One hawk
wins, while the other is injured. Since both hawks are equally
strong, the probability of winning or losing is
$\frac{1}{2}$. If a hawk meets a dove, the hawk wins and receives payoff $b$,
while the dove retreats and receives payoff $0$. If two doves meet, there will be no injury. One of them eventually wins.
The expected payoff is $\frac{b}{2}$. Thus the payoff matrix is given by
\begin{eqnarray*}
H \qquad D \qquad &\\
\\[-.3cm]
\mathcal{A}=\left(\begin{array}{ccc}
\dfrac{b-c}{2} & & b \\ 0 & & \dfrac{b}{2}
\end{array}\right).&\begin{array}{c} H \\[+.3cm] D \end{array}
\end{eqnarray*}
If $b<c$, then neither pure strategy is a Nash equilibrium. If everybody adopts the first pure strategy (H),
it is best to adopt the second pure strategy (D) and vice versa. This means that hawks and doves can coexist.
Selection dynamics will lead to a mixed population. \\

We fix $b=1$ and look for a suitable value of  $c>b$ for the numerical tests. We obtain the matrix
\begin{equation*}
\mathcal{A}=
\left(\begin{array}{ccc}
\dfrac{1-c}{2} & & 1 \\ 0 & & \dfrac{1}{2}
\end{array}\right).
\end{equation*}
In this case we have $-\dfrac{\beta}{\alpha}=\dfrac{1}{c}$ and $(\alpha, \beta)\in D$ if $c>1$ (see \figurename~\ref{fig:ABCD}). The function
\begin{equation}
f_0(p)=-p^2+ \theta p+1 \label{HD:dato2}
\end{equation}
is an admissible stationary solution for the problem, namely positive, with a total mass equal to 1, and with the first momentum equal to $-\dfrac{\beta}{\alpha}$, if and only if $\theta=\frac 2 3$ and $c=\frac{36}{17}$. In
\figurename~\ref{fig:testHD1} we show the numerical results for
the Cauchy problem (\ref{eq:model}), associated to this initial
datum. We remark that the numerical scheme
preserves the stationary solution.

\begin{figure}[htp]
\center \epsfig{file=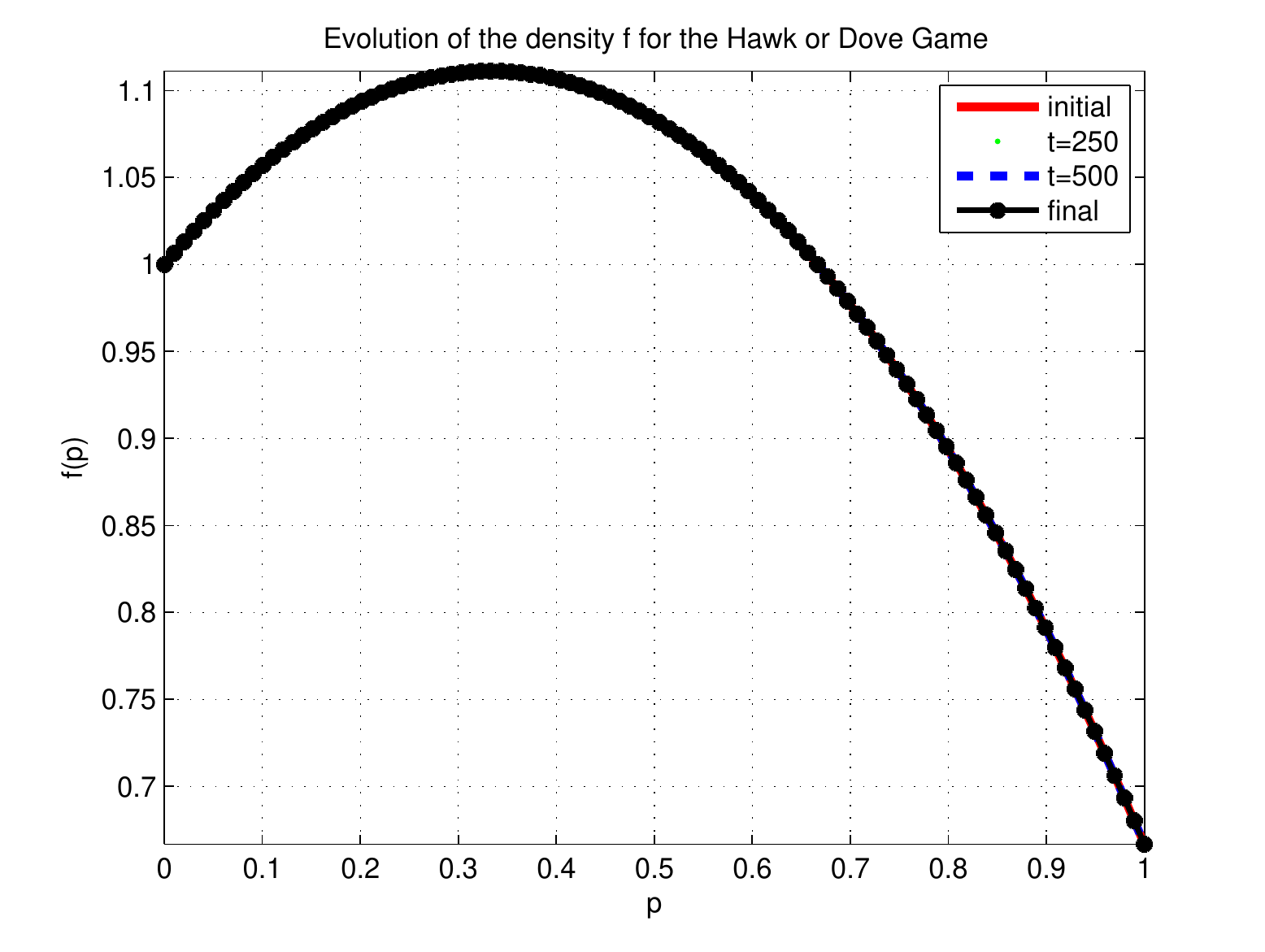,height=6cm,width=11cm} \caption{Hawk
or Dove Game, quadratic stationary solution: we plot the evolution
over time of $f(t,p)$ for the Cauchy problem (\ref{eq:model}) with
initial datum (\ref{HD:dato2}) for $T=1000$.} \label{fig:testHD1}
\end{figure}

As a consequence of Proposition \ref{prop:nostable}, the stationary
solution (\ref{HD:dato2}) is not linearly stable, in fact
\begin{equation*}
M_2(f_0(p))-(M_1(f_0(p)))^2=\int_0^1 \left( \left(-p^3+\dfrac{2}{3}p^2+p \right)\left(p-\dfrac{17}{36} \right) \right) dp\cong 0.07.
\end{equation*}
Actually, even small perturbations of the datum can generate large
perturbations on the solutions. We consider a perturbation with
zero mass for the function (\ref{HD:dato2}):
\begin{equation}\label{eq:pert2}
\tilde{f}_0(p)=-p^2+\dfrac{2}{3}p+1+0.02\sin(2\pi p),
\end{equation}
$\forall\, p\in [0,1]$. \figurename~\ref{fig:testHD3} shows the evolution of $f(t,p)$ for the related Cauchy problem.
\begin{figure}[htp]
\center \epsfig{file=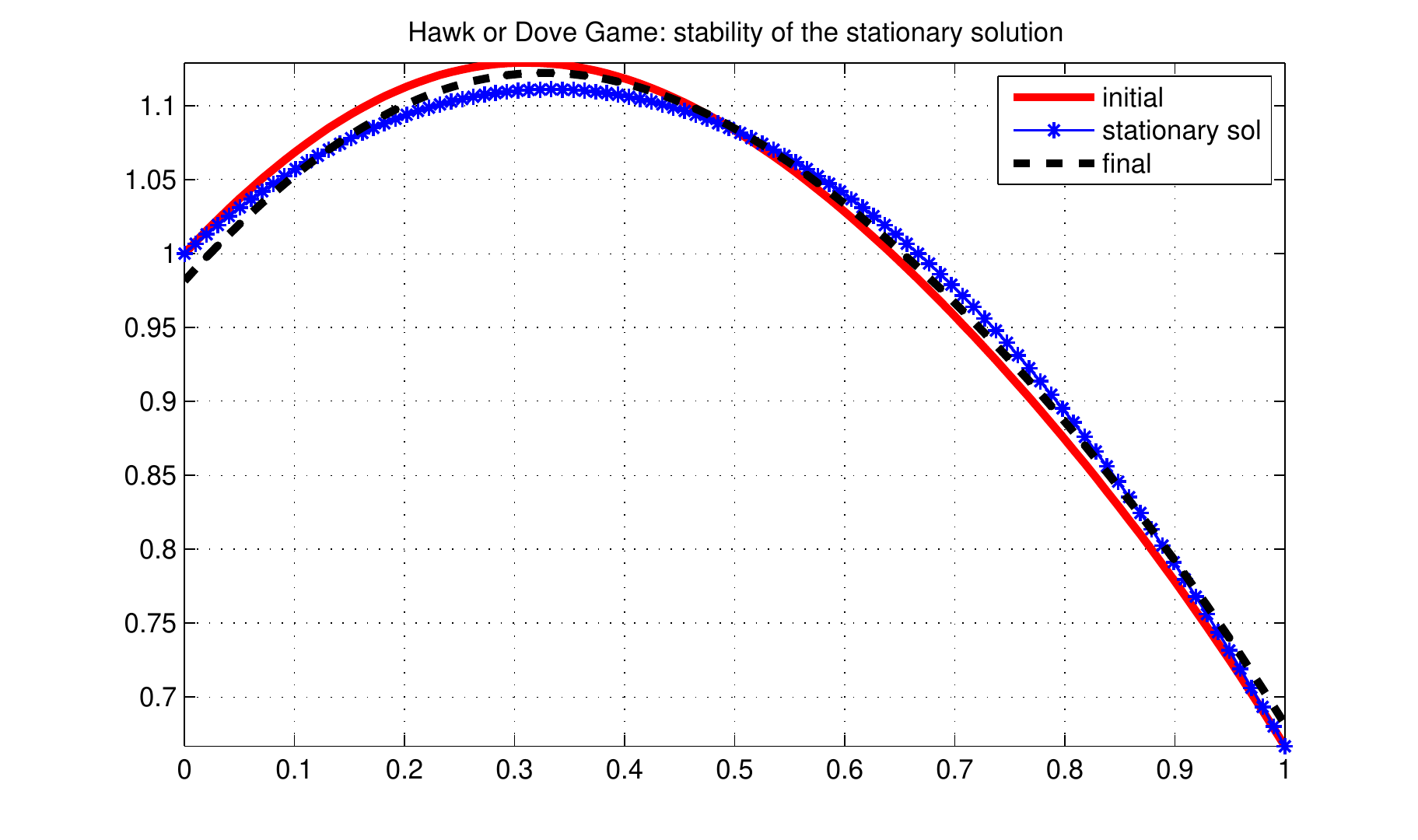,height=6cm,width=11cm}
\caption{Hawk or Dove Game, no stability of quadratic stationary
solution: we plot the evolution over time of $f(t,p)$ for the
Cauchy problem (\ref{eq:model}) with initial datum
(\ref{eq:pert2}) for $T=1000$.} \label{fig:testHD3}
\end{figure}
The perturbed datum (\ref{eq:pert2}) originates the loss of the
stationary solution, as seen in \figurename~\ref{fig:testHD3}. The
solution of the problem evolves (slowly) towards a Dirac mass and
we can see the first moment $M_1$ which converges to the value
$-\frac{\beta}{\alpha}=0.4722$ (\figurename~\ref{fig:HDcurva}), as
expected since $(\alpha,\beta)$ are in the region $D$.
\begin{figure}[htp]
\center \epsfig{file=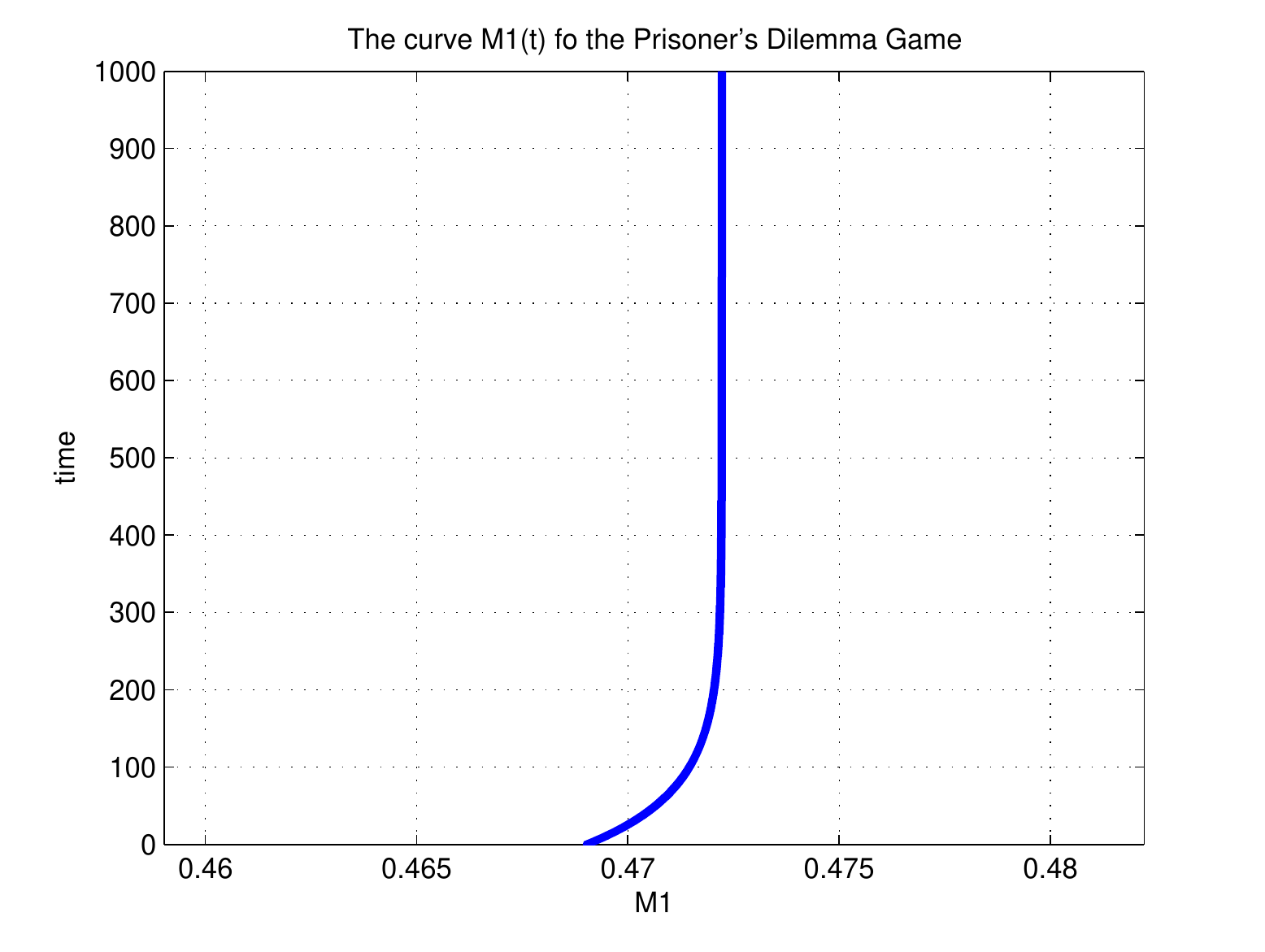,height=6cm,width=11cm}
\caption{Hawk or Dove Game, asymptotic behavior: we have $t$ on
the y-axes and $M_1(t)$ on the x-axes for the Cauchy problem
(\ref{eq:model}) with initial datum (\ref{eq:pert2}) for
$T=1000$.} \label{fig:HDcurva}
\end{figure}
\section{Three strategies games}
Assume there are three different strategies, whose interplay is ruled by the payoff matrix:
\[ \acal = \left(\begin{matrix} a_1 & a_2 & a_3\cr a_4 & a_5 & a_6 \cr a_7 & a_8 & a_9\end{matrix}\right).\]

We have a population where individuals are going to play strategy
A with probability $p_1$, strategy B with probability $p_2$ and
strategy C with probability $1-p_1-p_2$, for $(p_1,p_2)\in
\mathcal{T}_2$, where the simplex $\mathcal{T}_2$ is just
\begin{equation}\label{T2}
\mathcal{T}_2=\lbrace {\bf{p}}=(p_1,p_2)\in \R^2 \,|\, p_1,\, p_2 \geq 0,\, p_1+p_2\leq 1\rbrace.\end{equation}
The payoff (\ref{eq:payoff}) is given by
\begin{equation}\label{mix_payoff_3}
\begin{array}{ll}
A({\bf{p}},{\bf{p^*}})=&\left(\begin{array}{c} p_1 \\ p_2 \\ 1-p_1-p_2 \end{array}\right) \left(\begin{matrix} a_1 & a_2 & a_3\cr a_4 & a_5 & a_6 \cr a_7 & a_8 & a_9\end{matrix}\right)
\left(\begin{array}{c} p_1^* \\p_2^* \\ 1-p_1^*-p_2^* \end{array}\right)\\ \\
&=(a_1-a_3-a_7+a_9)p_1p_1^*+(a_2-a_3-a_8+a_9)p_1p_2^*\\ \\
&+(a_4-a_6-a_7+a_9)p_1^*p_2+(a_5-a_6-a_8+a_9)p_2p_2^*+(a_3-a_9)p_1\\ \\
&+(a_7-a_9)p_1^*+(a_6-a_9)p_2+(a_8-a_9)p_2^*+a_9\\ \\
&=\alpha p_1p_1^*+\beta p_1p_2^* +\gamma p_1^*p_2+\delta
p_2p_2^*+\sigma p_1+\eta p_1^*+\xi p_2+\mu p_2^*+\iota,\\ \\
\end{array}\end{equation} with $\alpha :=a_1-a_3-a_7+a_9$,
$\beta:=a_2-a_3-a_8+a_9$, $\gamma:=a_4-a_6-a_7+a_9$,
$\delta:=a_5-a_6-a_8+a_9$, $\sigma:=a_3-a_9$, $\eta:=a_7-a_9$,
$\xi:=a_6-a_9$, $\mu:=a_8-a_9$ and $\iota:=a_9$.
In this case, the problem (\ref{CauchyNdim}) is\\
\begin{equation} \label{Cauchy2dim}
\begin{cases}
\partial_t f({\bf{p}})=F(f) \quad t\geq 0,\, {\bf{p}}\in \mathcal{T}_2,\\
f(0,{\bf{p}})=f_0({\bf{p}}),
\end{cases}
\end{equation}
where the source term $F(f)$ is defined as follows:
\begin{eqnarray}\label{sorgente2D}
F(f)&:=&f({\bf{p}}) [(\alpha M_{(1,0)}(f)+\beta M_{(0,1)}(f)+\sigma)(p_1-M_{(1,0)}(f))\\
&+&(\gamma M_{(1,0)}(f)+\delta M_{(0,1)}(f)+\xi)(p_2-M_{(0,1)}(f))].\nonumber
\end{eqnarray}
We consider the initial datum $f_0({\bf{p}})$ such that $f_0({\bf{p}})\geq 0$ and $\int_{\mathcal{T}_2}f_0({\bf{p}})d{\bf{p}}=1$.
\begin{remark}\label{2Dstationary}
It is easy to prove that if
\begin{equation}\label{eq:hpstationary}
\dfrac{\beta \, \xi - \sigma \, \delta}{\delta \, \alpha - \gamma
\, \beta}>0 \qquad \mbox{and} \qquad \dfrac{\sigma \, \gamma - \xi
\, \alpha}{\delta \, \alpha - \gamma \, \beta}>0,
\end{equation}
then every distribution function $\bar{f}(\bf{p})$ with
$$M_{(1,0)}=\dfrac{\beta \, \xi - \sigma \, \delta}{\delta \,
\alpha - \gamma \, \beta} \qquad \mbox{and} \qquad
M_{(0,1)}=\dfrac{\sigma \, \gamma - \xi \, \alpha}{\delta \,
\alpha - \gamma \, \beta}$$ is a stationary solution for the
problem (\ref{Cauchy2dim}). Actually, by arguing as in Section 3,
also Dirac masses concentrated on points are stationary solutions
of these equations.
\end{remark}
\subsection{A special case: the Rock-Scissors-Paper Game}\label{RSP}
We consider the Rock-Scissors-Paper game, which is characterized
by having three pure strategies such that $R_1$ is beaten by
$R_2$, which is beaten by $R_3$, which is beaten by $R_1$. The
outcomes of the game are tabulated as
\begin{equation} \label{eq:payoffRSPgame}
\acal = \left(\begin{matrix} 0 & 1 & -1\cr -1 & 0 & 1 \cr 1 & -1 & 0\end{matrix}\right).
\end{equation}\\
In the Rock-Scissors-Paper game, the constants that appear in the source term (\ref{sorgente2D}) have the following values:
$\alpha=0$, $\beta=3$, $\sigma=-1$, $\gamma=-3$, $\delta=0$, $\xi=1$, and so
$$\dfrac{\beta \, \xi - \sigma \, \delta}{\delta \, \alpha - \gamma \, \beta}=\dfrac{\sigma \, \gamma - \xi \, \alpha}{\delta \, \alpha - \gamma \, \beta}=\dfrac{1}{3}.$$\\
The initial datum $f_0({\bf{p}})=2$ has integral over $\mathcal{T}_2$ equal to $1$ and the moments $M_{(1,0)}(2)=M_{(0,1)}(2)=\dfrac{1}{3}$
and so it is a stationary solution for this game. In the next Section \ref{2Dapproximation} we will present the numerical results related
to this stationary solution. \\

Now we want to present a result about the curve of changing sign for $\partial_t f$:
for this game the source term (\ref{sorgente2D}) is
$$F(f)=f\,[(3M_{(0,1)}(f)-1)p_1+(1-3M_{(1,0)}(f))p_2+M_{(1,0)}(f)-M_{(0,1)}(f)],$$
and so we have that the curve $\bar{p}(t)\subset \R^2$ such that $\partial_t f(\bar{p}(t))=0,$ is given by
\begin{equation}\label{eq:curveRSPgame}
p_1(3M_{(0,1)}-1)+p_2(1-3M_{(1,0)})+M_{(1,0)}-M_{(0,1)}=0,\qquad \forall\,(M_{(1,0)},M_{(0,1)})\in \mathcal{T}_2,
\end{equation}
that is the straight line joining the points $(M_{(1,0)}(f),
M_{(0,1)}(f))$ and $\left(\dfrac{1}{3}, \dfrac{1}{3}\right)$. In
the following Section \ref{2Dapproximation} we will present the
evolution over time of this straight line.

\section{Numerical approximation for the 3-strategies model}\label{2Dapproximation}
First we want to construct a numerical method for problem
(\ref{Cauchy2dim}). The domain $\mathcal{T}_2$ is just the
triangle with vertices $(0,0)$, $(1,0)$, $(0,1)$. In order to make
the numerical integrations, we fix a discretization step $\Delta p$ and a
uniform triangular grid in $\mathcal{T}_2$ as
\figurename~\ref{fig:T2} shows. Each point of the grid is
\begin{equation} \label{grid}
{\bf{p}}_{ij}:=(p_{1,i},p_{2,j}) \qquad i=0,\dots I,\,\,j=0,\dots, I-i,
\end{equation}
with $I:=\dfrac{1}{\Delta p}$. We use the notation
$g_{i,j}:=g(t,p_{1,i},p_{2,j})$ for all $i=0,\dots,I$ and
$j=0,\dots, I-i$ to indicate the value of a general function
$g(t,p_1,p_2)$ at each grid point ${\bf{p}}_{ij}$.
\begin{figure}[htp]
\center
\epsfig{file=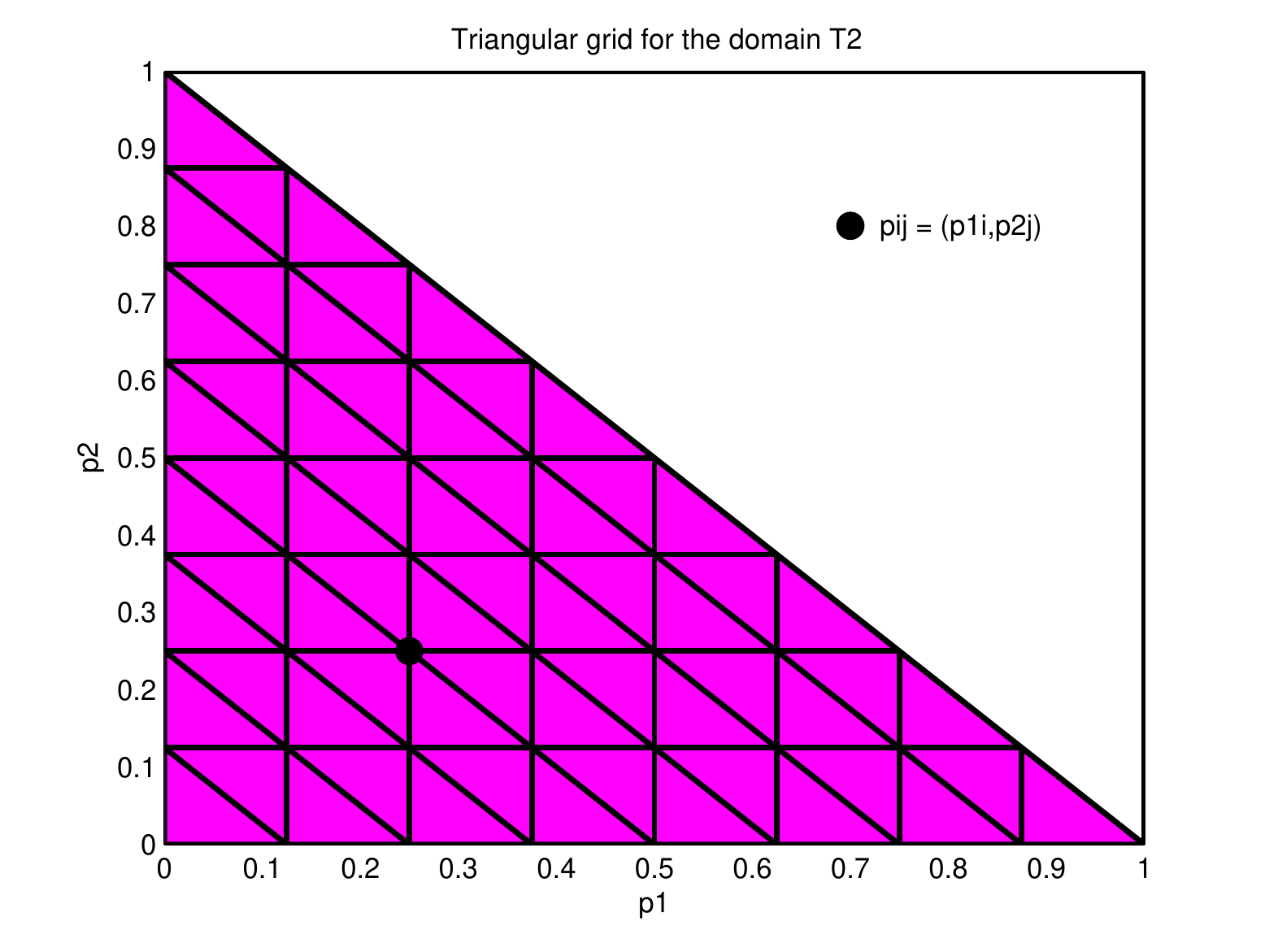, scale=0.6}
\caption{An example of triangular grid for the domain $\mathcal{T}_2$.} \label{fig:T2}
\end{figure}
In order to discretize the integral of $g$ over the domain
$\mathcal{T}_2$, we start to consider each triangle of the grid
and indicate its vertices as $(x_s,y_s)$, for $s=1,2,3$.
We define the following quantities:
\begin{eqnarray}\label{maxmin}
\overline{g}:=\max (g(x_s,y_s))\quad s=1,2,3 \\
\underline{g}:=\min (g(x_s,y_s))\quad s=1,2,3,
\end{eqnarray}
the maximum and the minimum value of $g$ on the triangle. On each
triangle of the grid we consider a $2D$ product formula based on the trapezoidal rule:
\begin{eqnarray}\label{eq:integrale}
&\int_{\mathcal{T}_2}& g(t,p_1,p_2)d{\bf{p}}=\int_0^1\,\int_0^{1-p_1} g(t,p_1,p_2)\,dp_2\,dp_1\nonumber \\ \nonumber \\
&\approx & \dfrac{\Delta p^2}{2}\left[\sum_{i=0}^{I-1}\sum_{j=0}^{I-i} \left(\underline{g}_1+\dfrac{1}{3}(\overline{g}_1-\underline{g}_1)\right)+\sum_{i=1}^{I-1}\sum_{j=0}^{I-i} \left(\underline{g}_2+\dfrac{1}{3}(\overline{g}_2-\underline{g}_2)\right)\right],
\end{eqnarray}
where
\begin{equation}\label{g1}
\overline{g}_1=\max(g_{i,j},g_{i,j+1},g_{i+1,j}), \quad \underline{g}_1=\min(g_{i,j},g_{i,j+1},g_{i+1,j}),
\end{equation}
\begin{equation}\label{g2}
\overline{g}_2=\max(g_{i,j},g_{i,j+1},g_{i-1,j+1}), \quad  \underline{g}_2=\min(g_{i,j},g_{i,j+1},g_{i-1,j+1}).
\end{equation}
The discretization of the first moments $M_{(1,0)}$ and
$M_{(0,1)}$ is easily obtained by (\ref{eq:integrale}),
considering the function $g(t,p_1,p_2)=p_1\,f(t,p_1,p_2)$ for
$M_{(1,0)}(f)$ and $g(t,p_1,p_2)=p_2\,f(t,p_1,p_2)$ for
$M_{(0,1)}(f)$. \\

Similarly to the one-dimensional case it can be shown that the method preserves
the total mass in time, as well as discrete analogous of the steady states.
As before the time discretization is done with a fourth order
Runge-Kutta method with constant time stepping.

\subsection{Numerical tests: The Rock-Scissors-Paper Game}
We consider the Cauchy problem associated to the problem (\ref{Cauchy2dim}) for the Rock-Scissors-Paper game
with the payoff matrix (\ref{eq:payoffRSPgame}). The problem has the following equation:\\
\begin{equation}\label{Cauchy:RSPgame}
\begin{cases}
\partial_t f({\bf{p}})=f({\bf{p}}) [(3M_{(0,1)}(f)-1)(p_1-M_{(1,0)}(f))\\
\quad \quad \quad \quad \quad +(1-3M_{(1,0)}(f))(p_2-M_{(0,1)}(f))],\\
f(0,{\bf{p}})=f_0({\bf{p}}).
\end{cases}
\end{equation}

\subsubsection{Test n.1} We start with an initial datum compactly supported in $\mathcal{T}_2$: the sum of $s\geq 1$ truncated
Gaussian functions, centered in $(p_1^{0,r},p_2^{0,r})$ for $r=1,\dots ,s$. To ensure that (\ref{massaf}) applies, we normalize
the datum so that its integral over $\mathcal{T}_2$ is equal to $1$. Our datum is of the following type
\begin{equation}\label{datum}
f_0({\bf{p}}):=\dfrac{G({\bf{p}})}{\int_{\mathcal{T}_2} G({\bf{p}})\,d{\bf{p}}},
\end{equation}
where
\begin{equation}\label{gaussian}
G({\bf{p}}):=\sum_{r=1}^s\,\max\left(\dfrac{1}{2\pi}\,e^{-K_r[(p_1-p_1^{0,r})^2+(p_2-p_2^{0,r})^2]}-0.01, 0\right),
\end{equation}
with $K_r>0$.

\begin{figure}[t]
\center \epsfig{file=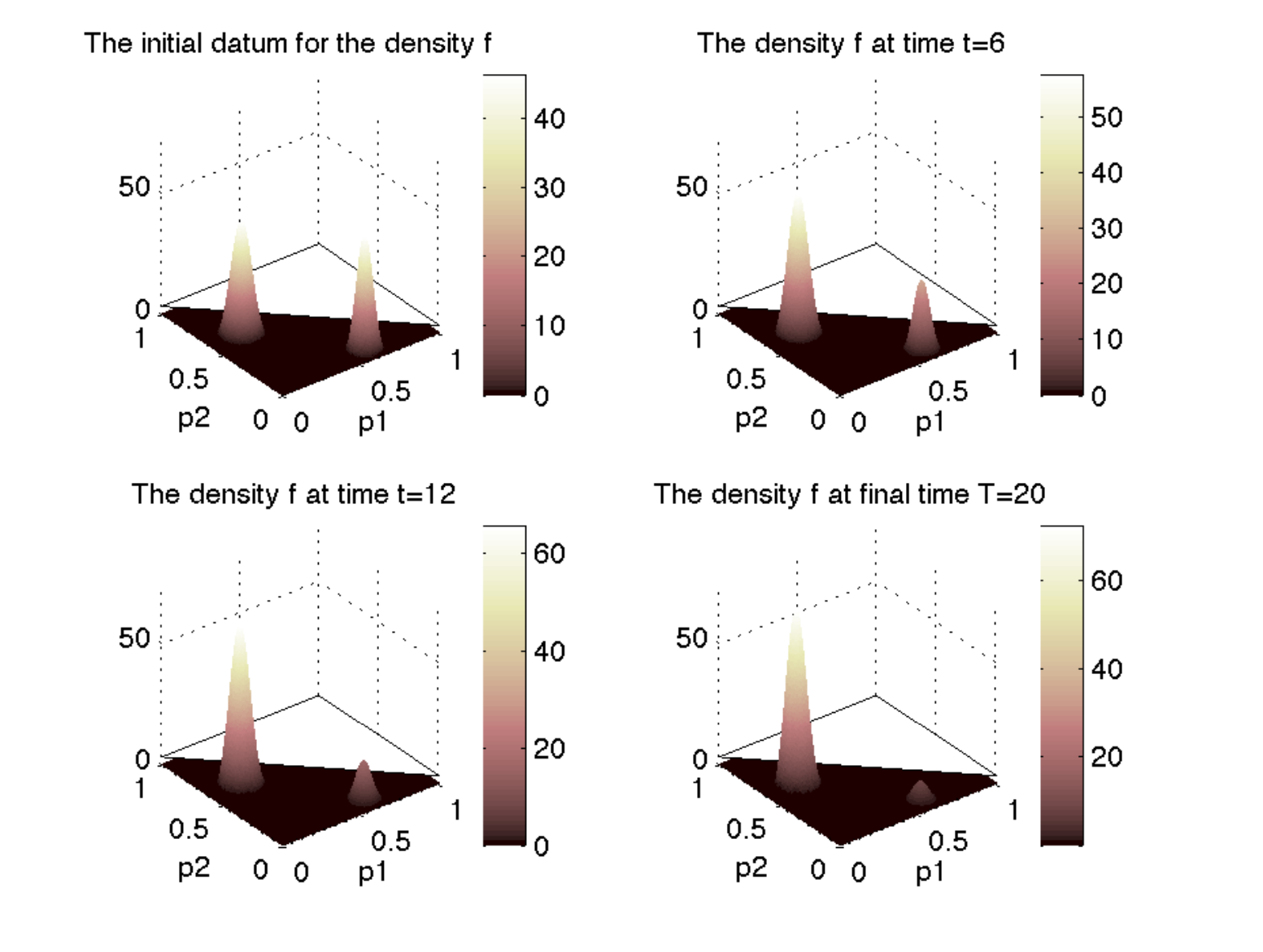,height=7.4cm,width=14cm}
\caption{Rock-Scissors-Paper Game, test 1.1: the evolution of the
density $f$ that is the numerical solution of the Cauchy problem
(\ref{Cauchy:RSPgame})-(\ref{datum}) with $s=2$,
$(p_1^{0,1},p_2^{0,1})=(\frac{1}{10},\frac{3}{5})$,
$(p_1^{0,2},p_2^{0,2})=(\frac{3}{5},\frac{2}{10})$, $K_1=300$,
$K_2=190$ and $T=20$.} \label{RSPcompact1}
\end{figure}

\subsubsection*{Test 1.1}
We fix $s=2$, $(p_1^{0,1},p_2^{0,1})=(\frac{1}{10},\frac{3}{5})$, $(p_1^{0,2},p_2^{0,2})=(\frac{3}{5},\frac{2}{10})$, $K_1=300$, $K_2=190$. \\\\
The graphical results (in \figurename~\ref{RSPcompact1}) shows that there is dominance of one of the groups: the initial datum is likely to
have two areas of concentration, the final configuration shows only one area of concentration, and the total $L^1$ mass remains constantly
equal to $1$ over time. The dominance group is contained in the region where $\partial_t f$ is positive as we can see in
the \figurename~\ref{RSPcurvesign1} that shows the contours of $f$ and the numerical results of the straight line $\bar{p}(t)$ of
changing sign for $\partial_t f$ (see Subsection \ref{RSP}). We also remark that, for all $t>0$, the support of $f(t)$ is equal or a subset
of the support of the initial datum $f_0$: $$supp (f(t))\subseteq supp (f_0),\qquad \forall\, t>0.$$

\begin{figure}[tp]
\center
\epsfig{file=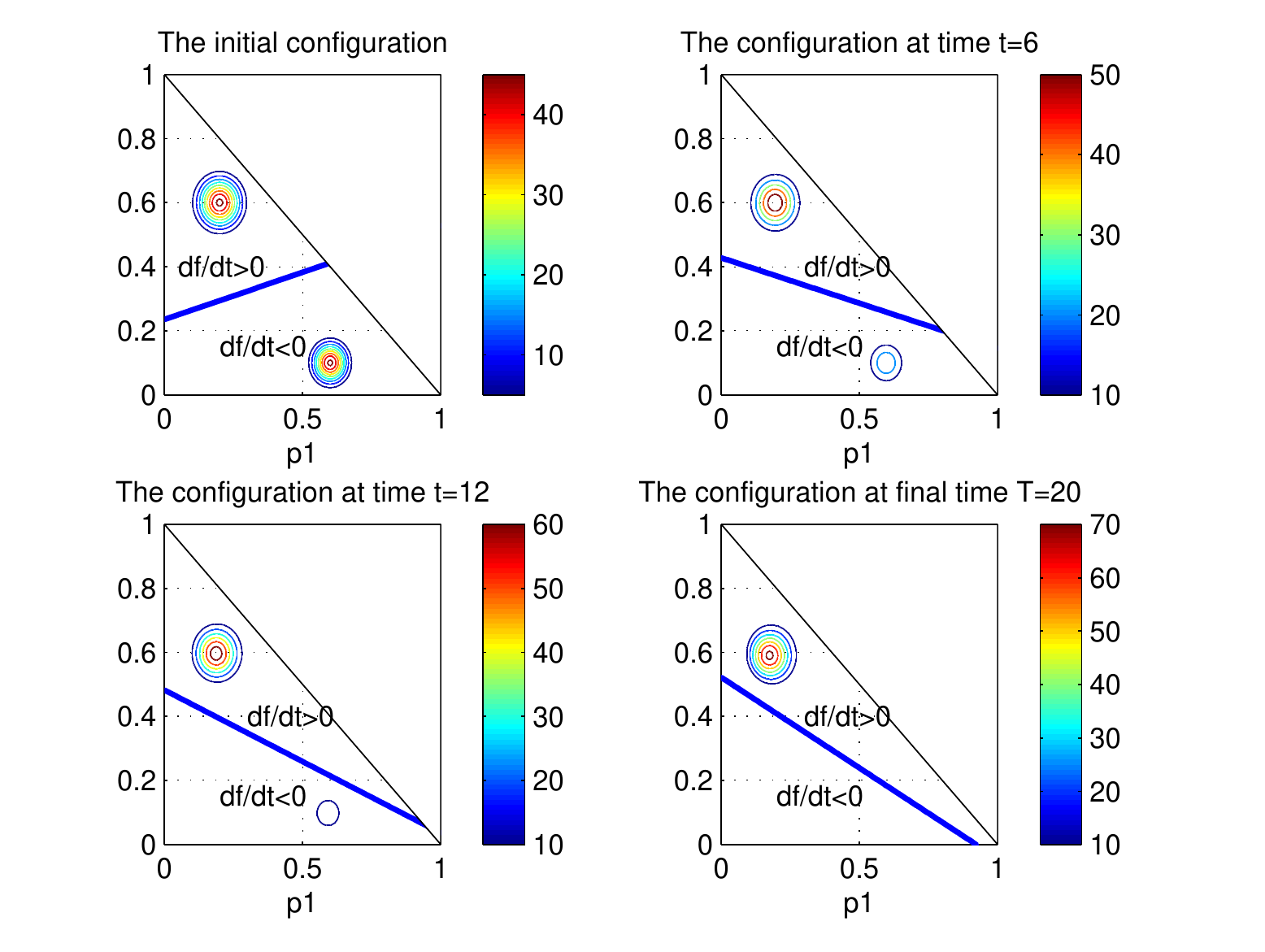,height=7.4cm,width=14cm}
\caption{Rock-Scissors-Paper Game, test 1.1: the evolution over time of the contours of the density $f$ and of the curve $\bar{p}(t)$ of change of sign for $\partial_t\,f$  for the Cauchy problem (\ref{Cauchy:RSPgame})-(\ref{datum}) with the same data and the same parameters as the previous \figurename~\ref{RSPcompact1}. } \label{RSPcurvesign1}
\end{figure}
\subsubsection*{Test 1.2} We fix $s=2$, $(p_1^{0,1},p_2^{0,1})=(\frac{1}{3},\frac{1}{3})$, $(p_1^{0,2},p_2^{0,2})=(\frac{3}{16},\frac{3}{16})$, $K_1=600$, $K_2=300$.\\\\
The graphical results (\figurename~\ref{RSPcompact2} and
\figurename~\ref{RSPcurvesign2}) show that the situation is
different from the previous test: in this case the initial datum
lies between the two regions where $\partial_t f$ is positive and
negative and the straight line $\bar{p}(t)$ of separation between
this two regions does not changes significantly over time. Therefore
the configuration of the function $f$ at the final time $T$ is not
very different from that at the initial time.

\begin{figure}[h]
\center \epsfig{file=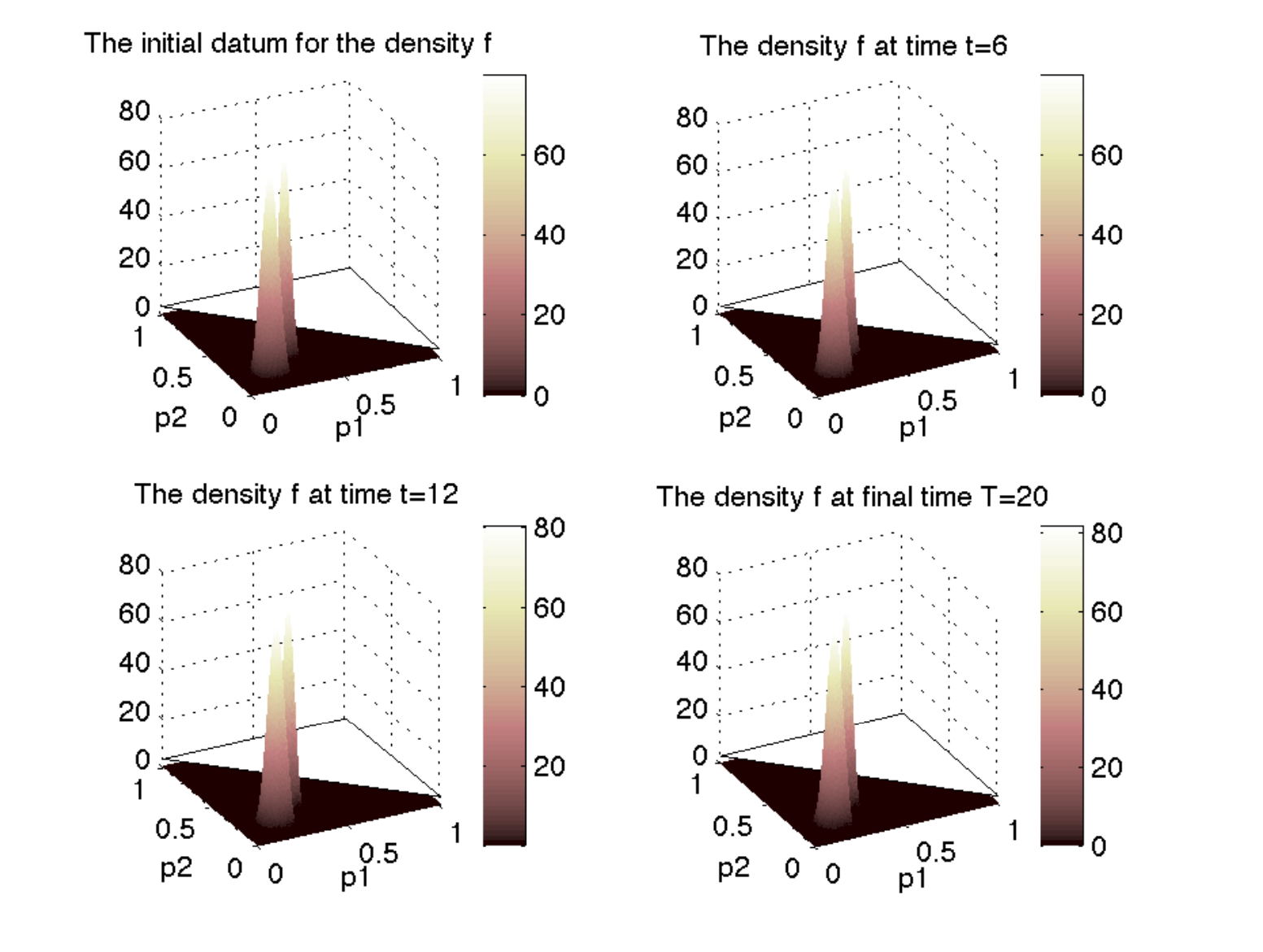,height=7.4cm,width=14cm}
\caption{Rock-Scissors-Paper Game, test 1.2: the evolution of the
density $f$ that is the numerical solution of the Cauchy problem
(\ref{Cauchy:RSPgame})-(\ref{datum}) with $s=2$,
$(p_1^{0,1},p_2^{0,1})=(\frac{1}{3},\frac{1}{3})$,
$(p_1^{0,2},p_2^{0,2})=(\frac{3}{16},\frac{3}{16})$, $K_1=600$,
$K_2=300$ and $T=20$.} \label{RSPcompact2}
\end{figure}
\begin{figure}[htp]
\center
\epsfig{file=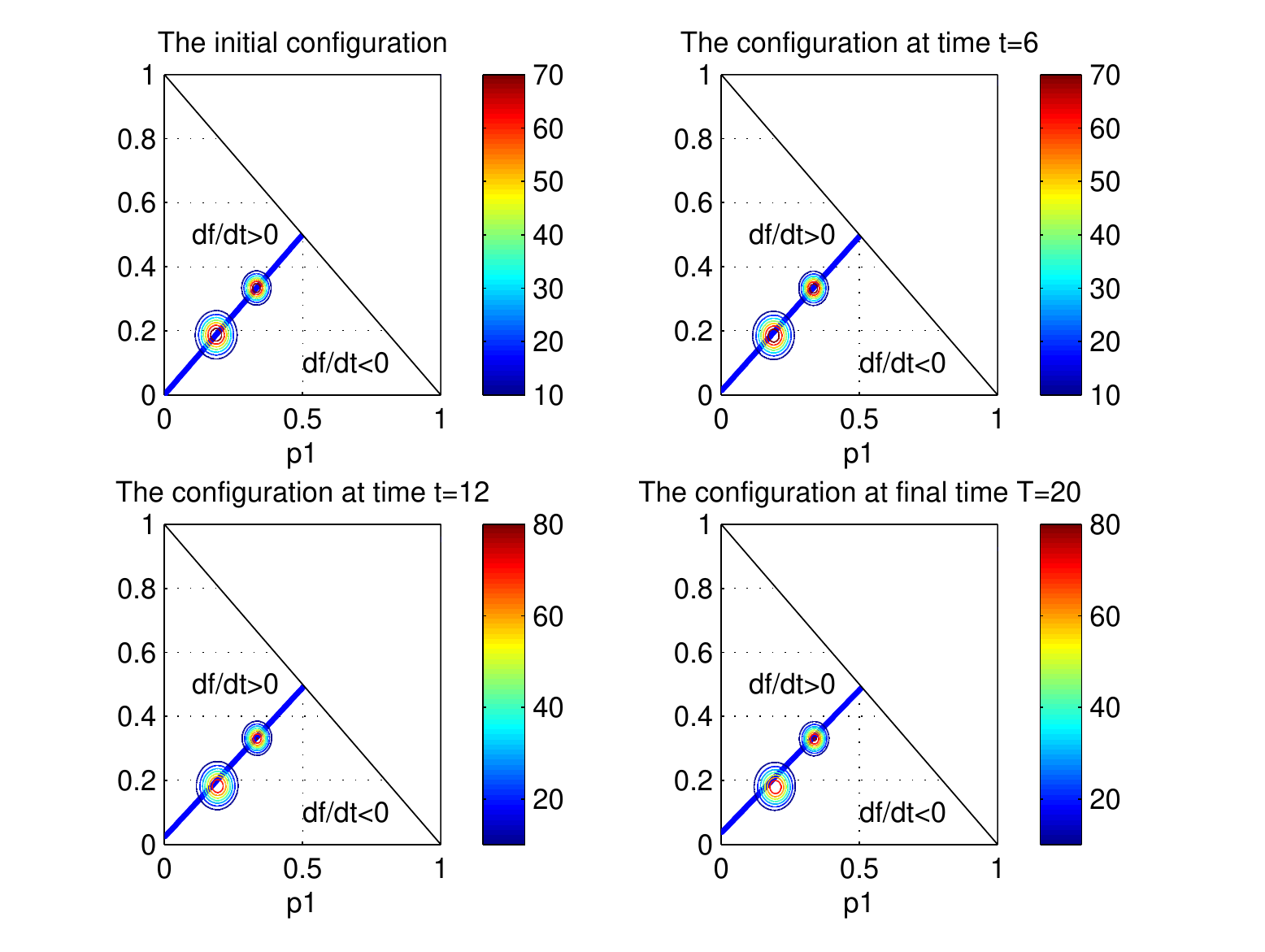,height=7.4cm,width=14cm}
\caption{Rock-Scissors-Paper Game, test 1.2: the evolution over
time of the contours of the density $f$ and of the curve $\bar{p}(t)$ of change of sign for $\partial_t\,f$ for the Cauchy problem
(\ref{Cauchy:RSPgame})-(\ref{datum}) with the same data and the same parameters as the previous \figurename~\ref{RSPcompact2}. } \label{RSPcurvesign2}
\end{figure}
\vspace{0.5cm}
\subsubsection{Test 1.3} We fix $s=3$, $(p_1^{0,1},p_2^{0,1})=(\frac{1}{3},\frac{1}{3})$, $(p_1^{0,2},p_2^{0,2})=(\frac{3}{16},\frac{3}{16})$, $(p_1^{0,3},p_2^{0,3})=(\frac{1}{10},\frac{3}{5})$, $K_1=600$, $K_2=300$ and $K_3=300$.\\\\
The graphical results (\figurename~\ref{RSPcompact3} and
\figurename~\ref{RSPcurvesign3}) show dominance phenomena in the
region where $\partial_t f$ is positive. Initially the three areas
of concentration are located, almost entirely, in the region where
$\partial_t f$ is negative. The time evolution shows us that,
already at $t=6$, two of the three areas of concentration are in
the middle between the two regions where $\partial_t f$ is
negative and positive, then to the final time, are completely in
the region where $\partial_t f$ is positive. So it is clear that
dominance takes place in these areas.

\begin{figure}[h]
\center \epsfig{file=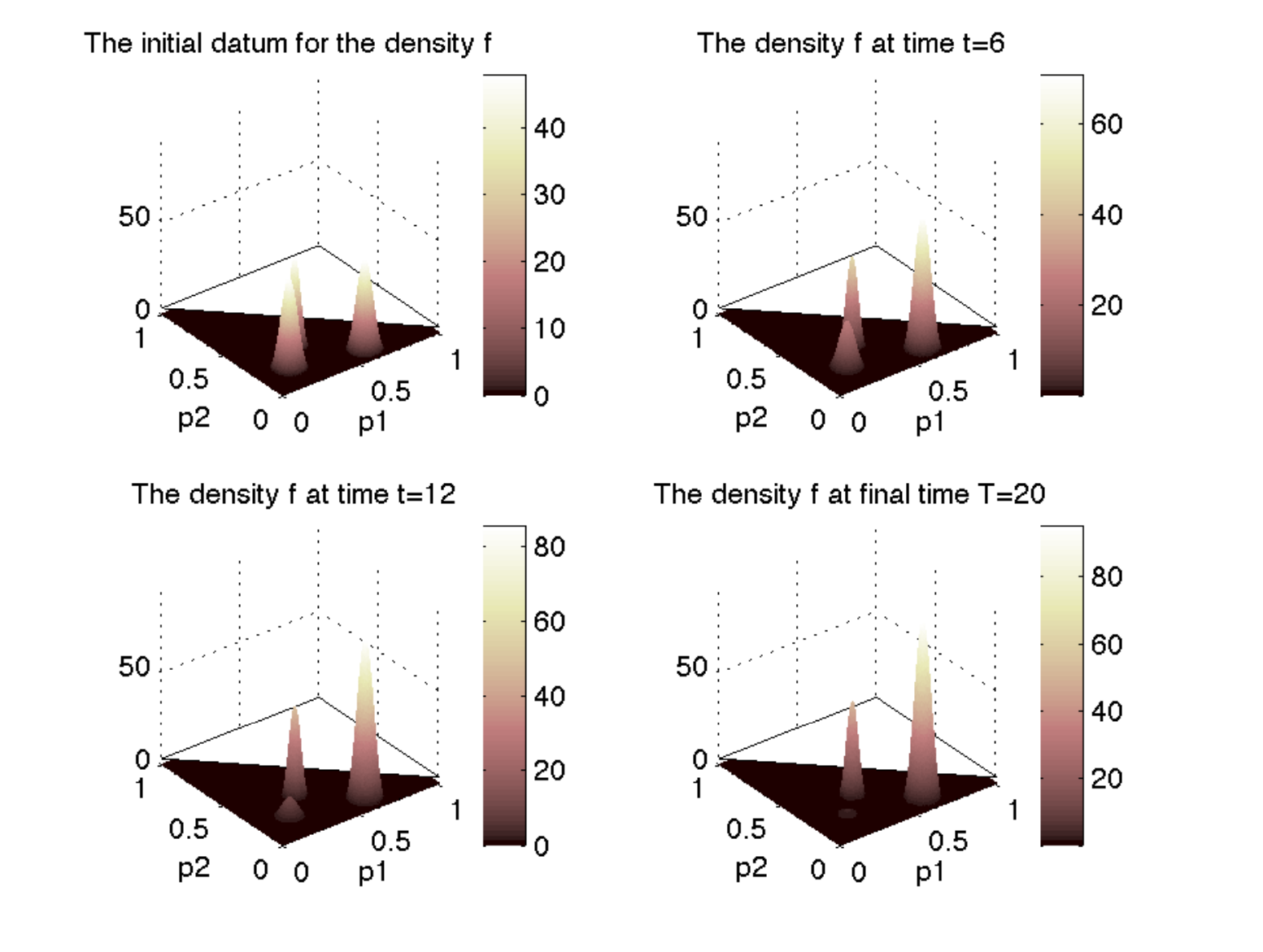,height=7.4cm,width=14cm}
\caption{Rock-Scissors-Paper Game, test 1.3: the evolution of the
density $f$ that is the numerical solution of the Cauchy problem
(\ref{Cauchy:RSPgame})-(\ref{datum}) with $s=3$,
$(p_1^{0,1},p_2^{0,1})=(\frac{1}{3},\frac{1}{3})$,
$(p_1^{0,2},p_2^{0,2})=(\frac{3}{16},\frac{3}{16})$,
$(p_1^{0,3},p_2^{0,3})=(\frac{1}{10},\frac{3}{5})$, $K_1=600$,
$K_2=300$, $K_3=300$ and $T=20$.} \label{RSPcompact3}
\end{figure}
\begin{figure}[htp]
\center
\epsfig{file=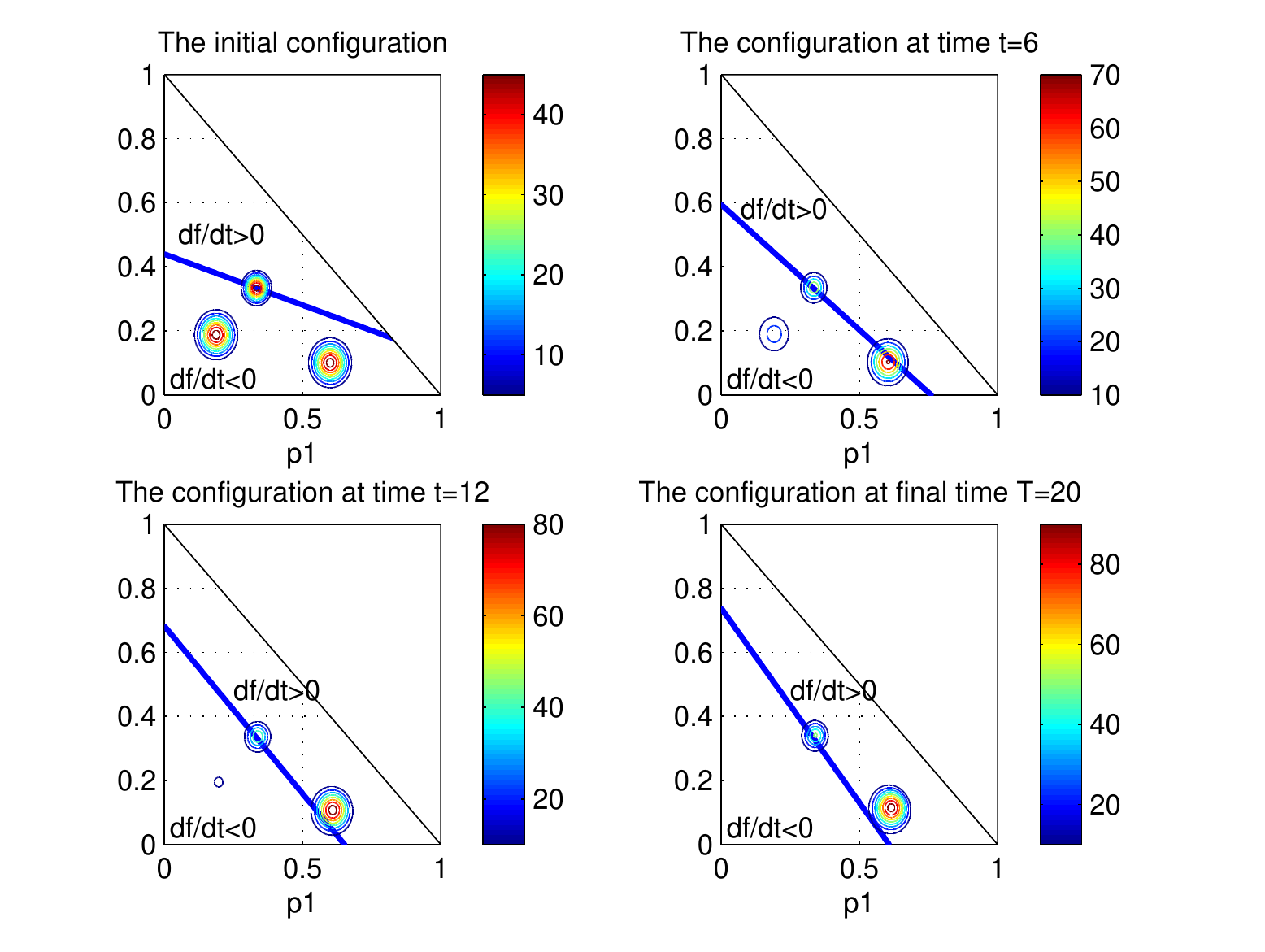,height=7.4cm,width=14cm}
\caption{Rock-Scissors-Paper Game, test 1.3: the evolution over
time of the contours of the density $f$ and of the curve $\bar{p}(t)$ of change of sign for $\partial_t\,f$ for the Cauchy problem
(\ref{Cauchy:RSPgame})-(\ref{datum}) with the same data and the same parameters as the previous \figurename~\ref{RSPcompact3}. } \label{RSPcurvesign3}
\end{figure}

\section{Conclusions}
We have considered a kinetic-like model for the evolution of a
continuous mixed strategy game. The model is based on the time
evolution of a density function describing the density of
population adopting a given strategy. We established several
analytical properties and develop some numerical discretizations
useful for numerical simulations in the case of two and three
strategies. Several explicit examples for two and three strategies
games are reported. Of course when considering more strategies a
deterministic approach may result in excessive computational
requirements and stochastic simulations methods should be
considered \cite{galstyan}.

Let us finally mention that, in the situation considered so far,
each player adopts a strategy and evolution over time leading to
survival or not of the player. In principle it can be interesting
to consider a situation in which each player can change strategy
by a random mutation, so moving through the strategy space. One
can introduce, to this end, a term in the equation that allows for the
random change of strategy, following the ideas presented in
\cite{volpert}. The most natural way to model this phenomenon is
to add a variation term in the equation (\ref{eq:fpmoments}), due
to the probability ${\bf{p}}\in \mathcal{T}_{N-1}$
\begin{equation}\label{diffusive}
\displaystyle \partial_t f-D\,\triangle_{{\bf{p}}}f=f\,\left(\sum_{i=1}^{N-1} (p_i-M_{{\bf{e}}_i}(f))\left(\upsilon_i+\sum_{j=1}^{N-1}\vartheta_{i,j}\,M_{{\bf{e}}_j}(f)\right)\right),
\end{equation}
with $D>0$. The new term $\triangle_{{\bf{p}}}f$ can be
interpreted as a diffusion term describing the spreading of the
population in the probability space from strategy to strategy,
which in evolution models corresponds to a random mutation
mechanism, and will be the object of a future work. A similar
model has been presented recently in \cite{ruijgrok}.



\begin{thebibliography}{99}
\bibitem{abrams} \newblock P. Abrams, \newblock {\em Modelling the adaptive dynamics of traits involved in inter- and intraspecific interactions: an
assessment of three methods}, \newblock Ecol. Lett. 4, (2001) 166--175.
\bibitem{bomze} \newblock I. Bomze, \newblock {\em Dynamical aspects of evolutionary stability},\newblock Mon. Math. 110, (1990) 189--206.
\bibitem{marsili} \newblock D. Challet, M. Marsili,
Y-C. Zhang, Minority Games: Interacting Agents in Financial Markets,  {\em Oxford University Press} (2005).
\bibitem{cressman} \newblock R.Cressman, {\em Stability of the replicator equation with continuous strategy
space}, \newblock Mathematical Social Sciences 50, (2005) 127--147.
\bibitem{desvillettes} \newblock L. Desvillettes, P.-E. Jabin, S. Mischler, G. Raoul, \newblock{\em On selection
dynamics for continuous structured populations}, \newblock  Commun. Math. Sci.  6,  (2008) 729--747.
\bibitem{friedman} \newblock D. Friedman, {\em Towards evolutionary game models
of financial markets}, \newblock Quantitative Finance 1, (2001)
\bibitem{galstyan}\newblock A. Galstyan, \newblock {\em Continuous Strategy Replicator Dynamics for
Multi--Agent Learning} (arXiv:0904.4717v1)
\bibitem{volpert} \newblock S. Genieys, N. Bessonov, V. Volpert, {\em Mathematical model of evolutionary
branching} \newblock Mathematical and Computer Modelling 49, (2009) 2109--2115.
\bibitem{wennberg} \newblock J. Henriksson, T. Lundh, B. Wennberg, \newblock {\em A model of sympatric speciation through reinforcement},
\newblock Kinet. Relat. Models  3,  (2010) 143--163.
\bibitem{hofbauer} \newblock J. Hofbauer, J. Oechssler, F. Riedel, {\em Brown-von Neumann-Nash dynamics: The continuous strategy
case}, \newblock Games and Economic Behavior 65, (2009) 406--429.
\bibitem{HofbauerSigmund} \newblock J. Hofbauer, K. Sigmund, {Evolutionary games and population dynamics},
\newblock {\em Cambridge University Press}, (1998)

\bibitem{norman} \newblock T.W.L. Norman, \newblock {\em Dynamically stable sets in infinite strategy
spaces}, \newblock Games and Economic Behavior 62, (2008) 610--627.
\bibitem{oechssler} \newblock J. Oechssler, F. Riedel, \newblock {\em Evolutionary dynamics on infinite strategy spaces},
\newblock Econ. Theory 17, (2001) 141--162.

\bibitem{ruijgrok} \newblock M. Ruijgrok, T. W. Ruijgrok, \newblock {\em Replicator dynamics with
mutations for games with a continuous strategy space},
(arXiv:nlin/0505032v2).

\bibitem{Weibull} \newblock J.W. Weibull, {Evolutionary game theory}, \newblock {\em MIT Press}, Cambridge, MA (1995).



\end{thebibliography}
\end{document}